\documentclass[sigconf]{acmart}
\usepackage{amsmath}
\usepackage{url}
\usepackage{xspace,balance,multirow}
\usepackage[font=bf]{caption}
\usepackage[ruled, vlined, linesnumbered]{algorithm2e}
\usepackage{xcolor}
\usepackage{soul}
\usepackage{colortbl}
\usepackage{bbold}
\usepackage{enumitem}
\usepackage[captionskip=-1pt]{subfig}
\usepackage{wrapfig}
\usepackage{flushend}
\usepackage{tikz}
\usepackage{pgfplots}
\usetikzlibrary{patterns}
\pgfplotsset{compat=1.16}
\usepackage{float}

\AtBeginDocument{%
  }

\copyrightyear{2026}
\acmYear{2026}
\setcopyright{cc}
\setcctype{by}
\acmConference[KDD '26]{Proceedings of the 32nd ACM SIGKDD Conference on Knowledge Discovery and Data Mining V.2}{August 09--13, 2026}{Jeju Island, Republic of Korea}
\acmBooktitle{Proceedings of the 32nd ACM SIGKDD Conference on Knowledge Discovery and Data Mining V.2 (KDD '26), August 09--13, 2026, Jeju Island, Republic of Korea}
\acmDOI{10.1145/3770855.3817920}
\acmISBN{979-8-4007-2259-2/2026/08}

\newcommand{\ie}{{i.e.,}\xspace}
\newcommand{\eg}{{e.g.,}\xspace}
\newcommand{\spara}[1]{\vspace{1mm}\noindent\textbf{#1.}}
\newcommand{\savespace}[1]{\ignorespaces}
\SetKw{KwAnd}{and}
\SetKw{KwOr}{or}
\SetKw{KwBreak}{break}
\SetKw{KwContinue}{continue}
\SetKw{KwTrue}{true}

\newcommand{\ours}{\textsf{PRA}\xspace}
\newcommand{\oursc}{\textsf{PRA$^c$}\xspace}
\newcommand{\oursa}{\textsf{APRA}\xspace}

\newcommand{\ATC}{\textsf{ATC}\xspace}
\newcommand{\ProFair}{\textsf{ProFair}\xspace}

\newcommand{\U}{\mathbf{U}\xspace}
\newcommand{\X}{\mathbf{X}\xspace}
\newcommand{\Z}{\mathbf{Z}\xspace}

\renewcommand{\b}{\mathbf{b}\xspace}

\newcommand{\x}{\mathbf{x}\xspace}
\newcommand{\y}{\mathbf{y}\xspace}

\newcommand{\z}{\mathbf{z}\xspace}
\newcommand{\I}{\mathcal{I}\xspace}

\newcommand{\W}{\mathcal{W}\xspace}

\newcommand{\C}{\mathcal{C}\xspace}

\newcommand{\R}{\mathbb{R}\xspace}
\renewcommand{\S}{\mathcal{S}\xspace}
\newcommand{\N}{\mathcal{N}\xspace}
\newcommand{\M}{\mathcal{M}\xspace}

\newcommand{\e}{{\ensuremath{\mathrm{e}}}}

\newcommand{\laks}[1]{{\textcolor{red}{}}} 

\newcommand{\LL}[1]{\textcolor{black}{#1}}

\newtheorem{definition}{Definition}
\newtheorem{theorem}{Theorem}
\newtheorem{lemma}{Lemma}
\newtheorem{claim}{Claim}
\newtheorem{proposition}{Proposition}

\newtheorem{corollary}{Corollary}[section]

\newcommand{\arxiv}[1]{{#1}}
\newcommand{\sketch}[1]{}

\newcommand{\eat}[1]{}
\newcommand{\cover}[1]{}

\newcommand{\hkk}[1]{}
\newcommand{\td}[1]{{\color{cyan}To do:}}
\newcommand{\Anson}[1]{{\color{green}Anson:}}

\def\addlegendimage{\csname pgfplots@addlegendimage\endcsname}

\pgfplotsset{every tick label/.append style={font=\tiny}}
\definecolor{myblue}{RGB}{113,210,242}
\definecolor{myorange}{RGB}{249,205,173}
\definecolor{myred}{RGB}{183,0,162}
\definecolor{mycolor}{RGB}{0,0,0}

\begin{document}

\title{Parameterized Fair Resource Allocation under Diversity Constraints}

\author{Keke Huang}
\authornote{Work partially done while at University of British Columbia as a postdoc.}
\affiliation{%
  \institution{Huazhong University of Science and Technology}
  \city{Wuhan}
  \state{Hubei}
  \country{China}
}
\email{kkhuang@hust.edu.cn}

\author{Yik Yu Ng}
\affiliation{%
  \institution{McGill University}
  \city{Montreal}
  \state{Quebec}
  \country{Canada}
}
\email{yik.ng@mail.mcgill.ca}

\author{Laks V.S. Lakshmanan}
\affiliation{%
  \institution{The University of British Columbia}
  \city{Vancouver}
  \state{BC}
  \country{Canada}
}
\email{laks@cs.ubc.ca}

\author{Xiaokui Xiao}
\affiliation{%
  \institution{National University of Singapore}
  \city{Singapore}
  \state{}
  \country{Singapore}
}
\email{xkxiao@nus.edu.sg}

\begin{abstract}

Resource allocation across multiple agent groups arises in many applications including e-commerce recommendation systems, housing assignment, and course allocation, and is commonly formulated as an optimization problem with diversity constraints to ensure group fairness. Existing approaches typically enforce these constraints as hard conditions, which overly restrict the feasible solution space and often lead to suboptimal allocations.

In this paper, we propose \ours, a parameterized framework for fair resource allocation under diversity constraints.  Inspired by the use of risk-aversion parameters in economic models, \ours introduces a set of controllable inequality-aversion parameters to softly regulate group-level diversity, thereby enabling flexible trade-offs between fairness and allocation efficiency. With appropriately calibrated parameters, \ours yields fairness-optimal assignments that comply with the specified diversity constraints. To accommodate additional application-specific constraints, we further extend the framework to an adaptive variant, \oursa. We establish that the optimality of both \ours and \oursa holds regardless of the chosen fairness metric and the nature of the additional constraints, underscoring the generality and robustness of our approach. Extensive experiments on three real-world applications demonstrate that our proposed framework consistently outperforms existing baselines in both effectiveness and robustness.
\end{abstract}

\ccsdesc[500]{Applied computing~Decision analysis}

\keywords{Resource allocation; Social Welfare; Fairness; Optimization}

\maketitle
\begin{sloppy}
\section{Introduction}\label{sec:intro}
Given a set of agents and a set of items, the problem of resource allocation is finding an allocation of items to agents in order to optimize a certain objective. Fair resource allocation or more broadly resource allocation under diversity constraints is a topic that has been extensively studied (e.g., see~\cite{steinhaus1948problem, BramsTaylor1996, Moulin17734, BrandtConitzer2016}) owing to its wide applications in government auctions~\cite{BarmanKV18}, social welfare allocation~\cite{RoosR10}, school course assignments~\cite{LouisNNS23}, and job recommendation~\cite{VladimirovaPD24}. In classic resource allocation settings, the agents often significantly outnumber the resources that need to be allocated to them. Agents are typically classified into different groups according to their attributes, e.g., gender, ethnicity, location, and items are categorized into different partitions based on their inherent properties, e.g., type, brand, etc.

\begin{figure}[t]
    \centering
    \begin{minipage}{\columnwidth}
        \centering
        \includegraphics[height=0.6\columnwidth]{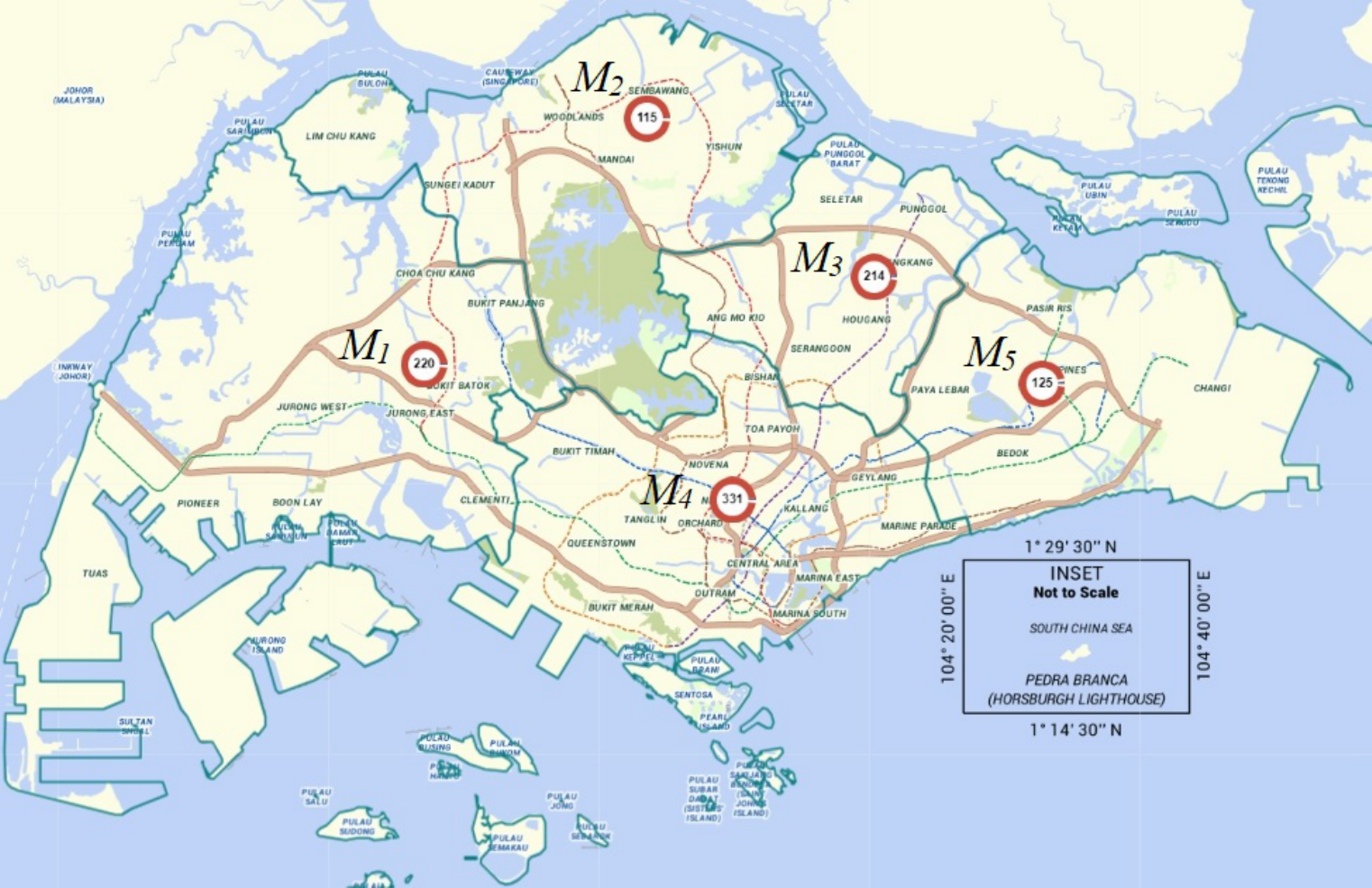}
    \end{minipage}
    \vspace{0mm}
    \begin{minipage}{\columnwidth}
        \footnotesize 
        \textbf{
        $\M_1:$ West Region (218 flats); $\M_2:$ North Region (114 flats); $\M_3:$ North-East Region (211 flats); $\M_4:$ Central Region (327 flats);  $\M_5:$ East Region (120 flats)        
        }
    \end{minipage}\vspace{-2mm}
\caption{The Distribution of HDB Flats in Singapore.}\label{fig:HDBdata} \vspace{-2mm}   
\end{figure}

One typical real-world application is the public housing allocation in Singapore~\cite{BenabbouCHSZ18,BenabbouCZ19}. Since 1989, the Singapore Housing and Development Board (HDB) has implemented an Ethnic Integration Policy (EIP)~\cite{singaporehdb1989} for  house allocation, aiming to accommodate its diverse ethnic and cultural population. As Singapore's public housing authority, HDB constructs government-subsidized public housing estates and sells them to Singapore residents. We collect up-to-date data on current HDB flats from the Singapore government's official website~\cite{hdbwebsite}
The HDB distribution and statistics are presented in Figure~\ref{fig:HDBdata}. According to the ethnic percentages of the population, the EIP establishes maximum limits on the proportion of flats in each estate, restricting occupancy to a maximum of 87\% Chinese, 25\% Malay, and 15\% Indian or other ethnic groups since March 2020~\cite{deng2013story,BenabbouCZ19}. \LL{This is an example of a \textit{diversity constraint}.} Another illustrative application is course assignments~\cite{partovi1995knowledge,LouisNNS23} in universities. Specifically, students from various departments compete for enrollment in popular public courses with limited capacity. Given the significant disparity in department sizes, ensuring group fairness is essential when selecting applicants across departments.

As agents belong to different groups, diversity constraints are typically enforced in a group-wise manner in resource allocation. In general, diversity constraints require that different groups are represented in a manner that is commensurate with their distribution in the population.  To achieve maximal fairness under such constraints, several types of approaches  have been proposed (details in Section~\ref{sec:related}). Among these approaches, two are particularly relevant. \citet{BenabbouCHSZ18,BenabbouCZ19} cast the problem as a linear programming problem and introduce \ATC, which aims to maximize the total utility under predefined capacity constraints to enforce diversity. While this formulation ensures feasibility within the specified bounds, it does not guarantee fairness-optimal solutions, as multiple feasible allocations with varying fairness levels can exist \cite{BenabbouCHSZ18}. Furthermore, \ATC treats each agent and item as a distinct optimization variable, resulting in substantial memory consumption and an increased risk of out-of-memory (OOM) issues, as evidenced by our experimental results (Section~\ref{sec:jobresults}). In contrast, \ProFair~\cite{LouisNNS23} models diversity using explicit upper and lower bounds on group allocations, thereby directly controlling the distribution proportions. However, the rigidity of these hard constraints limits the flexibility of the optimization process and often results in suboptimal fairness. 

To address the lack of flexibility of the hard constraints, we propose \ours, a \underline{P}arameterized framework for \underline{R}esource \underline{A}llocation. To this end, we investigate the {\em risk-aversion} mechanism as formulated in economic models~\cite{HeidariFGK18}, where parameterized formulations are employed to characterize risk preferences ranging from aversion to seeking in social welfare analysis. Inspired by this, we introduce {\em inequality-aversion} parameters to formalize the utility function governing allocation among groups. \LL{Inequality metrics are often hard to directly optimize. We get around this difficulty by a novel means. Suppose we have a certain desirable allocation that we want to target. We show that we can always find inequality-aversion parameters such that the allocation achieving maximum welfare w.r.t.\ those aversion parameters is the original target allocation. This is a powerful result as it allows us to reach any desirable target allocation as the welfare-maximizing allocation w.r.t.\ carefully chosen aversion parameters. The target allocation may be one that has low or zero Atkinson inequality, possibly subject to supplementary constraints besides diversity, or has indeed low inequality w.r.t.\ some other inequality metric altogether, such as {\em statistical parity}~\cite{HertweckHL21}, {\em gini index}~\cite{Farris10,charles2022gini}, or {\em nash welfare}~\cite{kaneko1979nash,CaragiannisKMPS19}. We complement these results by showing that a simple greedy algorithm \ours based on marginal gain leads to an efficient solution for finding allocations for maximum welfare, thanks to concave welfare functions.} By tweaking the inequality-aversion parameters, we prove that \ours can attain \LL{highly flexible and} varying degrees of fairness in resource allocation. Meanwhile, we employ the {\em Atkinson inequality}~\cite{atkinson1970measurement} as a \LL{main} measure of group fairness. By relaxing the allocation problem to a continuous setting with infinitesimally divisible resources, we theoretically establish that the continuous version of \ours, termed \oursc, minimizes the inequality score via appropriately calibrated parameters, leading to improved group fairness. When applied to the discrete case of indivisible items, we provide a theoretical bound on the inequality gap relative to the continuous optimum.

In addition to diversity constraints, practical applications often impose supplementary constraints. To address these, we extend \ours to an adaptive variant, \oursa. We further show that the fairness optimality of \ours and \oursa holds for a broad class of inequality metrics, and that welfare optimality is preserved even in the presence of additional application-specific box constraints. The experimental results across three real-world applications demonstrate the effectiveness and robustness of our framework.

In a nutshell, our contributions are as follows.
\begin{itemize}[topsep=1mm,partopsep=0pt,itemsep=1mm,leftmargin=18pt]
\item We propose \ours, a parameterized framework for resource allocation under diversity constraints. 
It adopts a flexible inequality-aversion parameter mechanism instead of hard constraints and is capable of offering maximum welfare while achieving minimum inequality regardless of the underlying fairness metrics (Sections~\ref{sec:optimization},~\ref{sec:generalInequality}).
\item We study the inequality bounds arising from allocations with practical indivisible items. Specifically, we provide bounds on the deviation from the divisible optimum (Section~\ref{sec:gaps}).
\item We further extend \ours to an adaptive variant, \oursa, to accommodate additional application-specific constraints. \oursa achieves welfare-optimal performance while adhering to the \LL{supplementary} constraints (Section~\ref{sec:ExtCon}).
\item We evaluate the framework across three real-world applications. The superiority of \ours and \oursa over the baselines bears testimony to   the effectiveness of our framework for resource allocations (Section~\ref{sec:exp}).
\end{itemize}

All formal proofs are provided in Appendix~\ref{sec:proofs}.\vspace{-0.5mm}

\section{Preliminaries}\label{sec:pre}

\subsection{Notations and Definitions}\label{sec:notation}

We use bold uppercase letters, bold lowercase letters, and letters in calligraphic fonts respectively to represent matrices (\eg $\X$), vectors (\eg $\x$), and sets (\eg $\S$). For a positive integer $n$,  $[n]=:\{1,2,\cdots, n\}$. 

Let $\N$ be a set of agents with $|\N|=n$ and $\M$  a set of items (\eg goods, houses) with $|\M|=m$ for $n,m \in \Z_+$. Due to the scarcity of resources, it normally holds that $n\gg m$, \ie the number of agents is significantly larger than the number of items. Agents and items are typically categorized into distinct groups \LL{or partitions} based on the inherent attributes of agents and items such as gender, ethnicity, location, and brand. In particular, $\N$ contains $K$ groups, \ie $\N=\bigcup^K_{k=1}\N_k$, and $\M$ consists of $L$ partitions of various sizes, \ie $\M=\bigcup^L_{\ell=1}\M_\ell$. Without loss of generality, we assume agent groups (resp. item partitions) are non-empty and pairwise disjoint.  
 
We use matrix $\X\in \{0,1\}^{n\times m}$ to indicate an allocation:  $\X[i,j]=1$ if agent $i\in \N$ is assigned item $j\in \M$; otherwise $\X[i,j]=0$. Frequently used notations are summarized in Table~\ref{tbl:notations}.

\begin{table}[!ht]
    \centering
    \caption{Frequently used notations}\label{tbl:notations}
    \vspace{-2mm}
    \setlength{\tabcolsep}{0.2em} 
    \renewcommand{\arraystretch}{1}
    \renewcommand{\aboverulesep}{0pt}
    \renewcommand{\belowrulesep}{0pt}
	\begin{tabular}{@{}c|m{6.5cm}@{}}\toprule 
        \textbf{Notation} & \multicolumn{1}{c}{\textbf{Description}} \\ \midrule
        $\N, \N_k$ & the set of agents and the $k$-th group\\ \midrule
        $\M, \M_\ell$ & the set of items and the $\ell$-th partition \\ \midrule        
        $n,m$ & the number of agents in $\N$ and number of items in $\M$  \\ \midrule
        $K,L$ & the number of groups in $\N$ and number of partitions in $\M$\\ \midrule
        $\X$ &  indicator matrix $\X\in \{0,1\}^{n\times m}$ for an allocation \\ \midrule
        $\U$ & matrix $\U \in \R^{K\times L}$ indicates the allocation of agent groups on item partitions \\ \midrule
        $\alpha$ & inequality-aversion parameter for agent groups $\alpha\in \R^K$ \\ \midrule
        $\W(\U,\alpha)$ & the welfare function \\  \bottomrule
    \end{tabular}
    \vspace{-1mm}
\end{table}

\subsection{Individual Utility and Social Welfare}\label{sec:utiwel}

In applications of resource allocation, {\em individual utility} and {\em social welfare}  are the two fundamental optimization objectives~\cite{dolan1998}. Utility functions are commonly used to measure an individual's satisfaction and well-being in relation to specific items. It quantifies the happiness or satisfaction an individual derives from consuming various products or services. In contrast, welfare functions are concerned with the well-being of the entire collection of agents, i.e., the society. In particular, it assesses the overall well-being or utility by considering the collective interests of all individuals within a society. Therefore, it primarily focuses on the {\em fairness} of resource distribution among individuals \LL{while having efficiency of allocation as a prerequisite}. In reality, simply maximizing the individual utility over society often results in unequal resource distribution across diverse social groups. This imbalance leads to unfairness and undermines the overall social welfare~\cite{dolan1998, JoeWongSLC12,BenabbouCHSZ18}. 

Therefore, resource allocation among multiple groups in real-world applications typically involves inherent  constraints specific to each group. \LL{Conventional research \cite{LouisNNS23} has taken the view that} the allocation for any group must neither fall below nor exceed predefined thresholds. These unique thresholds, specific to each group, are known as \emph{diversity constraints}.

 \subsection{Problem Definition}\label{sec:probdef}

We let matrix $\U \in \R^{K\times L}$ record the unit utility\footnote{We omit the variations in individual utility as we focus on group fairness.} of agents of each of the  $K$ groups over  items in the $L$ partitions. \LL{In particular, the group utility of an allocation of items in partition $\ell$ to agents in group $k$ is given by $\U[k, \ell]=\sum_{i\in \N_k}\sum_{j\in \M_\ell}\X[i,j]$ for $k\in [K]$ and $\ell\in [L]$. We follow the literature on resource allocation and assume that the number of agents in any group is significantly larger than the number of items, so it follows that $\sum_{\ell\in[L]}\U[k,\ell] < |\N_k|$ for $k\in[K]$.} Meanwhile, we assume $|\M_\ell|>K$ such that each group is assigned at least one item, \ie $\U[k,\ell]\ge 1$ for $k\in[K], \ell\in [L]$ (details in Section~\ref{sec:alpha}, Corollary~\ref{cor:nonempty}). 

Diversity constraints and the resource allocation problem under diversity constraints are formally defined as follows. 
\begin{definition}[Diversity Constraints]\label{def:divcons}
Consider an agent set $\N$ of $K$ groups and an item set $\M$ of $L$ partitions. Let $\U[k,\ell]\in \mathbb{N}$ be the number of items from the $\ell$-th partition assigned to the $k$-th group. The diversity constraint requires that \LL{for given constants} $0 < \lambda_1 \le \lambda_2$,  for all $k \in [K]$ and $\ell \in [L]$, the allocation satisfies
\[\textstyle\lambda_1 \cdot \frac{|\N_k|}{|\N|} \le \frac{\U[k,\ell]}{|\M_\ell|} \le \lambda_2 \cdot \frac{|\N_k|}{|\N|} .\]
\end{definition}

Intuitively, the constraints ensure that the utility of a group relative to the size of an item partition is ``close'' to the relative size of the group in the population. The parameters $\lambda_1, \lambda_2$ help us control just how close the two ratios need to be.  We assume that the parameters $\lambda_1,\lambda_2$ are chosen such that the resulting constraints admit feasible integer allocations.

\begin{definition}[Resource Allocation under Diversity Constraints]\label{def:probdef}
Consider an agent set $\N$ of $K$ groups and an item set $\M$ of $L$ partitions. The resource allocation problem aims to compute an allocation utility matrix $\U \in \R^{K\times L}$ that satisfies the diversity constraints.
\end{definition}

Notice that the problem aims to find a utility matrix $U$ without specifying the actual allocation $X$ corresponding to $U$. That is, the problem does not distinguish between individual agents in a group. This flexibility means that any allocation $X$ compatible with a given utility $U$ is an acceptable solution. As such, we focus on determining the utility matrix $U$ that optimizes certain desirable objectives. As well, there may exist multiple feasible allocation matrices, each potentially resulting in different levels of welfare. To promote group fairness in welfare, we resort to the concept of  {\em risk-aversion} from  economics~\cite{HeidariFGK18}, where parameterized models are used to capture a spectrum of risk preferences ranging from \LL{risk-aversion to risk-seeking}. Building on this, we introduce {\em inequality-aversion} parameters to formalize the utility function that govern inter-group allocations. In particular, we introduce a vector of parameters $\alpha \in (0,1)^K$ to compute welfare from a given utility matrix. Formally, we define the welfare function $\W: \R^{K \times L} \times \R^K \to \R$ as follows.

\begin{definition}[Welfare function]\label{def:welfare}
Given a utility matrix $\U \in \R^{K\times L}$ and an inequality-aversion vector $\alpha \in (0,1)^K$, the welfare function is defined as $\W(\U, \alpha)=\sum_{\ell\in [L]}\sum_{k\in [K]} \big(\tfrac{\U[k,\ell]}{|\M_\ell|}\big)^{\alpha_k}$. 
\end{definition}\vspace{-1mm}

The inequality-aversion parameters $\alpha$ for different groups quantitatively measure the level of welfare obtained from the utility. Given $\tfrac{\U[k,\ell]}{|\M_\ell|} \in (0,1)$, smaller inequality-aversion parameters amplify the contribution of the utility to overall welfare, thereby improving fairness. In view of this, minority groups may favor smaller values of $\alpha$, whereas majority groups usually prefer larger values.

\section{Framework with Optimal Welfare}\label{sec:framework}

\subsection{Optimal Welfare under Diversity Constraints}\label{sec:optimization}

Definition~\ref{def:welfare} defines overall welfare from group utilities by leveraging the inequality-aversion parameter $\alpha$. By calibrating parameter $\alpha$, it is possible to ensure the resource allocation among groups satisfies given diversity constraints when aiming to maximize welfare (details in Section~\ref{sec:atkinson} and~\ref{sec:alpha}). 
As such, we can attain the desired resource allocation. 

Based on this, given the parameter $\alpha\in (0,1)^K$, we propose a principled framework for {\em welfare optimization under diversity constraints} as follows. 
\begin{align}
\max\  &\W(\U,\alpha), \notag \\
\mathrm{s.t.} &\textstyle \sum_{i\in \N} \X[i,j] \le 1,\  \forall j\in \M,  \notag \\
&\textstyle \sum_{j\in \M} \X[i,j] \le 1,\ \forall i\in \N,\\
&\textstyle \X[i,j] \in \{0,1\},\ \forall i\in \N, j\in \M,  \notag
\end{align}
where $\U[k, \ell]=\sum_{i\in \N_k}\sum_{j\in \M_\ell}\X[i,j]$ for $k\in [K]$ and $\ell\in [L]$. \LL{The constraints ensure that no agent is allocated more than one item, and no item is allocated to more than one agent.} Meanwhile, the objective $\max\W(\U,\alpha)$ inherently guarantees the assignment of all items to agent groups since there are more agents than items.

As informally argued earlier, the diversity constraints are incorporated implicitly in the objective function, which are fulfilled automatically when the allocation achieves the maximum welfare. To this end, we propose a Parameterized Resource Allocation algorithm \ours (pseudo-code in Algorithm~\ref{alg:RA}) and prove that \ours can achieve optimality, thus yielding the maximum welfare. Specifically, the core idea of \ours is to assign an item from each partition to the group yielding the largest marginal gain in welfare. When multiple groups provide the same largest marginal gain, we break ties by uniformly sampling one of them. Since all such maximizers are equivalent, this is without loss of generality and can be treated as the single-maximizer case. Hence, we ignore this distinction in the remainder of the paper.

\begin{algorithm}[!t]
\begin{small}
\caption{Resource Allocation \ours}\label{alg:RA}
\KwIn{Agent set $\N$, Item set $\M$, inequality-aversion parameter $\alpha$, group number $K$, partition number $L$}
\KwOut{Utility matrix $\U$}
Initialize $\U\gets \{0\}^{K\times L}$\;
\For{$\ell \gets 1$ to $L$}
{
  \For{$t \gets 1$ to $|\M_\ell|$}
  {
     $\S^\ast \gets \arg\max_{k\in [K]}\left(\big(\tfrac{U[k,\ell]+1}{|M_\ell|}\big)^{\alpha_k}-\big(\tfrac{U[k,\ell]}{|M_\ell|}\big)^{\alpha_k}\right)$\;
     $k^\ast \sim \mathrm{Uniform}(\S^\ast)$\;
     $\U[k^\ast,\ell]\gets \U[k^\ast,\ell]+1$\;
  }
}
\Return $\U$\;
\end{small}
\end{algorithm}

\spara{Marginal gain function} Recall that vector $\U[\cdot, \ell]\in \mathbb{N}^K_+$ denotes the number of  items allocated to the $K$ groups from the $\ell$-th item partition. Accordingly, the corresponding contribution of the $k$-th group to the welfare,  $\W(\U[k,\ell],\alpha_k)$,  is calculated as $\W(\U[k,\ell],\alpha_k)=\big(\tfrac{\U[k,\ell]}{|\M_\ell|}\big)^{\alpha_k}$. Intuitively, $\W(0,\alpha_k)=0$ holds for $k\in[K]$.
For ease of exposition,  define the group-wise {\em marginal gain function} $g$ as 
\begin{equation}\label{eqn:margaingain}
\textstyle g(\U[k,\ell],\alpha_k)=\W(\U[k,\ell],\alpha_k)-\W(\U[k,\ell]-1,\alpha_k),   
\end{equation}
for $\U[k,\ell]\ge 1$ for $k\in[K], \ell\in [L]$ (see Corollary~\ref{cor:nonempty}). According to the property of the concave function, we have the following straightforward proposition. 
\begin{proposition}\label{pro:monotone}
The marginal gain function $g(\cdot,\cdot)$ is monotonically decreasing w.r.t. the number of allocated items, i.e., $g(\U[k,\ell]+1,\alpha_k) \leq g(\U[k,\ell],\alpha_k), \forall k,\ell$. 
\end{proposition}

Utility vector $U[\cdot,\ell]$ on the $\ell$-th partition is constructed via $|\M_\ell|$ iterations in Algorithm~\ref{alg:RA} by leveraging the greedy strategy.
We establish the following lemma.

\begin{lemma}\label{lem:marginal}
Consider the utility matrix $\U$ from an arbitrary iteration for item allocation on $\M_\ell$ in Algorithm~\ref{alg:RA} for $\ell\in [L]$. It holds that $g(\U[j,\ell],\alpha_j) \ge g(\U[i,\ell]+1, \alpha_i)$ for $\U[j,\ell]\ge 1$ for $\forall i,j\in [K]$.
\end{lemma}

We then prove that Algorithm~\ref{alg:RA} achieves the optimal welfare.

\begin{theorem}\label{thm:optimalG}
Given the input parameter $\alpha \in (0,1)^K$, let $\U$ be the utility matrix computed by Algorithm~\ref{alg:RA}. Given $|\N_k| > \sum_{\ell\in[L]}\U[k,\ell] $ for $k\in[K]$, the welfare $\W(\U, \alpha)$ of  is optimal.
\end{theorem}

Theorem~\ref{thm:optimalG} establishes the fact that when the inequality-aversion parameter $\alpha$ is specified (by Algorithm~\ref{alg:alphas} in Section~\ref{sec:alpha}), the utility allocation $\U$ output by \ours yields the optimal welfare $\W(\U, \alpha)$. Moreover, an appropriate selection of $\alpha$ ensures the allocation $\U$ adheres to the diversity constraints imposed on each group. 

\spara{Item assignment within a group} Algorithm~\ref{alg:RA} produces a group-level resource allocation distribution without committing to which agents in a group are assigned the items. By standard conventions~\cite{GrossHumbertBB23, dolan1998},  items ``assigned'' to a group can be actually assigned to agents with higher utility. This intra-group allocation process is orthogonal to the problem studied in this paper but can be effectively handled by existing methods~\cite{BenabbouCHSZ18, LouisNNS23, BenabbouCZ19}.  

\subsection{Atkinson Inequality on Fairness}\label{sec:atkinson}

As discussed in Section~\ref{sec:utiwel}, social welfare is normally measured by {\em fairness} in resource allocation across agent groups, which is often quantified using inequality measures. Among these, the {\em Atkinson inequality} metric~\cite{atkinson1970measurement} serves as a foundational tool widely used in economics. Accordingly, we adopt the Atkinson inequality as the main metric for our analysis.

\begin{definition}[Atkinson inequality~\cite{atkinson1970measurement}]
Consider a vector $\b\in \R_{\ge 0}^K$ representing the vector of utilities of $K$ groups from an allocation and a parameter $\beta \in (0,1)$. The Atkinson inequality metric $A(\b, \beta)$ is defined as
\begin{equation}\label{eqn:atkinson}
\textstyle A(\b, \beta) = 1-\tfrac{1}{\mu}\left(\tfrac{1}{K}\textstyle\sum_{i=1}^{K}\b^{(1-\beta)}_i\right)^{1/(1-\beta)},
\end{equation}
where $\mu=\tfrac{1}{K}\textstyle\sum_{i=1}^{K}\b_i$ is the averaged utility.
\end{definition}

It can be shown using Jensen’s inequality that the Atkinson inequality metric exhibits the desirable property that $A(\b, \beta)\ge 0$ and $A(\b, \beta)=0$ if and only if $\b_1=\b_2=\cdots=\b_K$.        

Essentially, the Atkinson inequality assesses the imbalance of the allocation distribution among groups. In particular, the parameter $\beta$ governs the sensitivity of the Atkinson inequality metric to the distributional imbalance represented by the vector $\mathbf{b}$. The above  \LL{desirable properties}   hold irrespective of the choice of $\beta$. Meanwhile, recall that the \LL{group level distribution computed by Algorithm~\ref{alg:RA}} intrinsically relies on the parameter $\alpha$. \LL{This raises the question, how to determine $\alpha$ to achieve a desired level of fairness, which we further explore in the subsequent section.} 

\subsection{Determination of $\alpha$ under Atkinson Inequality}\label{sec:alpha}

Throughout this section, we consider the divisible-item relaxation of \ours, termed \oursc. In particular, \oursc is interpreted as a continuous greedy allocator that repeatedly assigns an infinitesimal amount of resource to the group with the largest marginal welfare gain. We note that \oursc is solely for the purpose of analysis. In the following, we elaborate on how to determine the appropriate parameter $\alpha \in (0,1)^K$ in terms of the Atkinson Inequality. 

Without loss of generality, we focus on a specific arbitrary partition $\M_\ell$ of the item set $\M$. Let $\x \in (0,1)^K$ denote the vector representing the proportion of assignments of items from partition $\M_\ell$ to the $K$ groups. In this case, the  optimization problem on partition $\M_\ell$ can be  formalized as 
\begin{align}\label{eqn:concave}
& \textstyle\max_{\x \in (0,1)^K} \sum_{k=1}^K \x_k^{\alpha_k},  \\
& \mathrm{s.t.} \textstyle \sum_{i\in [K]} \x_k =1,  \x_k > 0,\notag
\end{align}
which is a concave optimization problem. Thus, we have the following lemma.

\begin{lemma}[\cite{BV2014}]\label{lem:kkt}
For any $\alpha \in (0,1)^K$, the objective function $\sum_{k=1}^K\x_k^{\alpha_k}$  over $\x \in (0,1)^K$ subject to $\sum_{k\in [K]} \x_k =1$ is strictly concave and admits a unique global maximum, characterized by the KKT condition $\alpha_k \x_k^{\alpha_k-1} = \lambda$ for $k\in [K]$ for some constant $\lambda$.
\end{lemma}

Edmonds’ Greedy Theorem~\cite{schrijver2003combinatorial} points out that the greedy strategy achieves the optimal solution for the concave optimization problem in Equation~\eqref{eqn:concave}. Therefore, upon the termination of the allocation in \oursc, it returns the optimal allocation $\x$ which ensures the following marginal equality condition.
\begin{equation}\label{eqn:marginalequation}
\alpha_1 \x_1^{\alpha_1-1}=\alpha_2 \x_2^{\alpha_2-1}=\cdots=\alpha_K \x_K^{\alpha_K-1},
\end{equation}
where $\x_1+\x_2+\cdots \x_K=1$ and $\alpha_k, \x_k\in (0,1)$ for $k\in [K]$. Without loss of generality, we assume $\x_1 \ge \x_2\ge \cdots \ge \x_K$ by reindexing the pairs $(\x_k, \alpha_k)$ for $k\in [K]$. We establish the following theorem.

\begin{theorem}\label{thm:generalalpha}
Given any target allocation $\x = (\x_1,\ldots,\x_K) \in (0,1)^K$ such that $\sum_{k=1}^K \x_k = 1$, 
there exists a parameter vector $\alpha = (\alpha_1,\ldots,\alpha_K) \in (0,1)^K$ such that
$\x$ satisfies the marginal equality condition in Eq.~\eqref{eqn:marginalequation}.
\end{theorem}

To prove Theorem~\ref{thm:generalalpha}, we first explore the pairs $(\x_1, \alpha_1)$ and $(\x_2,\alpha_2)$, as formalized in the following lemma.

\begin{lemma}\label{lem:fixportion}
Given $\x_1,\x_2 \in (0,1)$ with $\x_1\ge \x_2$ and $\x_1+\x_2\in (0,1]$, and initialized $\alpha_1\in (0,1)$, there exists $\alpha_2\in(0,\alpha_1]$ such that $\alpha_1 \x_1^{\alpha_1-1}=\alpha_2 \x_2^{\alpha_2-1}$ holds.
\end{lemma} 

Given an initial value of $\alpha_1$ and the targeted allocation $\x$, the constant $\lambda$ is calculated as $\lambda=\alpha_1 \x_1^{\alpha_1-1}$. $\alpha_k$ in Equation~\eqref{eqn:marginalequation} is successively calculated by solving $\alpha_k \x_k^{\alpha_k-1} = \lambda$ for $k\in \{2,\cdots,K\}$. By the same continuity arguments as in the proof of Lemma~\ref{lem:fixportion}, there exists a solution $\alpha_k\in (0,\alpha_{k-1}]$.

\spara{Determining $\alpha$ in terms of the Atkinson Inequality} Theorem~\ref{thm:generalalpha} shows that for a suitable choice of $\alpha$, the desired allocation $\x$ that minimizes the Atkinson inequality can be obtained by $\oursc$. In what follows, we then discuss how to determine such $\alpha$ in terms of the Atkinson Inequality.

According to the property of the Atkinson inequality, the inequality is minimized when all elements in the vector $\b$ are equal.  In  resource allocation, vector $\b$ is the vector of average utilities of the $K$ groups, \ie $\b_k=\tfrac{\x_k}{|\N_k|}$ for $k\in [K]$. Therefore, 
\begin{equation}\label{eqn:zeroatkinson}
\tfrac{\x_1}{|\N_1|}=\cdots=\tfrac{\x_k}{|\N_K|}    
\end{equation}
leads to the minimum inequality. By  setting $t=\tfrac{\x_1}{\x_2}\in [1,\infty)$, $\alpha_1 \x_1^{\alpha_1-1}=\alpha_2 \x_2^{\alpha_2-1}$ becomes equivalent to
\begin{equation}\label{eqn:atkinsongeneral}
\textstyle \alpha_1(1+\tfrac{1}{t})^{1-\alpha_1}r^{\alpha_1-\alpha_2}=\alpha_2(1+t)^{1-\alpha_2}.   
\end{equation} 
where $r=\x_1+\x_2$. Observe that the value of $r$ depends on $\x_1 +\x_2$, which cannot be determined in advance. In addition, to satisfy the condition $\frac{\x_1}{\x_2} = \frac{|\mathcal{N}_1|}{|\mathcal{N}_2|}$ as required for minimizing  the Atkinson Inequality, we  initialize $r=\frac{|\mathcal{N}_1|}{|\mathcal{N}|} + \frac{|\mathcal{N}_2|}{|\mathcal{N}|}$ at the outset.

Based on this, we propose Algorithm~\ref{alg:alphas} to compute $\alpha$. W.l.o.g., we assume  $|\N_1| \ge |\N_2| \cdots \ge |\N_K|$, which implies $\x_1\ge \x_2\ge \cdots \ge \x_K$. As there is no closed form of $\alpha_2$, we approximate $\alpha_2$ in an iterative manner with a sufficiently small input stride $\epsilon$ in Algorithm~\ref{alg:alphas}. 
Empirically, we set $\epsilon$ on the order of $1/\max\{|\mathcal{N}_1|, \ldots, |\mathcal{N}_K|\}$. For $\alpha_1$, we set $\alpha_1=1/M_{max}$ where $M_{max}=\max\{|\M_1|, \cdots, |\M_L|\}$.

\begin{algorithm}[!t]
\begin{small}
\caption{Determination of $\alpha$}\label{alg:alphas}
\KwIn{Agent set $\N$, group number $K$, parameter $\alpha_1$, stride parameter $\epsilon$}
\KwOut{$\alpha$}
Sort $\N$ by $|\N_k|$ in a non-increasing order for $k\in[K]$\;
\For{$k \gets 2$ to $K$}
{
   $r\gets \tfrac{|\N_{k-1}|+|\N_k|}{|\N|}$, $t\gets \tfrac{|\N_{k-1}|}{|\N_k|}$\;   
   $\alpha_k \gets \alpha_{k-1}$\;
  \While{$\alpha_k \ge 2\epsilon$}
  {
     \If{$\alpha_{k-1}(1+\tfrac{1}{t})^{1-\alpha_{k-1}}r^{(\alpha_{k-1}-\alpha_k)}<\alpha_k(1+t)^{1-\alpha_k}$}
     {
       $\alpha_k \gets \alpha_k -\epsilon$
     }
     \Else 
     {\KwBreak\;}     
  }
}
\Return $\alpha$\;
\end{small}
\end{algorithm}

Algorithms~\ref{alg:RA} and ~\ref{alg:alphas} together distribute items within each partition across groups. \LL{Throughout this allocation process, group fairness is our focus.} Once the group allocation matrix $\U$ is determined by \ours, individual utilities can be subsequently incorporated within each element $\U[k,\ell]$, for all $k \in [K]$ and $\ell \in [L]$. For instance, as described in~\cite{GrossHumbertBB23}, upon determining a group allocation, each group then distributes items to its members in a manner that optimizes utilitarian social welfare~\cite{GrossHumbertBB23, dolan1998}.

According to the property of Atkinson Inequality and Algorithm~\ref{alg:alphas}, we establish the following non-empty property of the allocation by \ours as follows.
\begin{corollary}\label{cor:nonempty} Consider the allocation $\U$ from \ours with the parameter $\alpha$ derived by Algorithm~\ref{alg:alphas} with initialization $\alpha_1=1/M_{max}$ where $M_{max}=\max\{|\M_1|, \cdots, |\M_L|\}$. When $|\M_\ell|>K$ holds for $\ell\in [L]$, we have $\U[k,\ell]\ge 1$ for $k\in[K], \ell\in [L]$. 
\end{corollary}

\spara{Nash Welfare Metric} Nash welfare (NW)~\cite{BenabbouCIZ20,Kell023,GuptaNCVRNDC23} is also a commonly used metric to measure the overall welfare of the allocation distribution. In particular, it quantifies the welfare by calculating the geometric mean of utilities among groups. The corresponding group Nash welfare on all blocks is calculated as 
\begin{equation}\label{eqn:nash}
\textstyle\textrm{NW}=\sum_{\ell \in [L]}\left(\prod_{k\in[K]} \left(\tfrac{U[k,\ell]}{|\N_k|}\right)^{|\N_k|}\right)^{1/|\N|}.   
\end{equation}
As indicated in Section~\ref{sec:utiwel}, Nash welfare also prefers fair allocation among groups~\cite{CaragiannisKMPS19}.  

\subsection{Optimality Gap Under Indivisible Items}\label{sec:gaps}

In real-world applications where items being allocated are indivisible, inherent discrepancies arise between the proportions of allocated items to groups and the ideal proportions specified in Equation~\eqref{eqn:marginalequation}. In this section, we show that the gap between the ideal allocation and that obtained by \ours is bounded. 

\begin{lemma}\label{lem:similarmargin}
Let $\U$ be the utility matrix returned from \ours and $\x=(\x_1,\x_2,\cdots,\x_K)\in(0,1)^K$ be an allocation satisfying Equations~\eqref{eqn:marginalequation} and~\eqref{eqn:zeroatkinson}. It holds that $\U[k,\ell]\in [\max\{\lfloor \x_k|\M_\ell| \rfloor -K, 0\} +1, \lceil \x_k|\M_\ell| \rceil +K-1]$ for $k\in [K]$ and $\ell\in [L]$. 
\end{lemma}

Since $\x$ satisfies Equations~\eqref{eqn:marginalequation} and ~\eqref{eqn:zeroatkinson}, it enjoys zero Atkinson inequality. What can we say about the inequality incurred by the (group level) allocation coming from \ours, given the above gap? We next establish a bound on the inequality.

\begin{theorem}\label{thm:worstinequality}
Consider a partition $\M_\ell$ with $M=|\M_\ell|$, the allocation $\U$ returned by \ours on $\M_\ell$ across $K$ groups, the divisible optimum $\x=(\x_1,\ldots,\x_K)\in(0,1)^K$ satisfying Equations~\eqref{eqn:marginalequation} and~\eqref{eqn:zeroatkinson}, and any $\beta\in(0,1)$. Without loss of generality, assume $|\N_1|\le \cdots \le |\N_K|$. The Atkinson inequality of $\U$ on $\M_\ell$ is upper bounded by
\[\textstyle 1-\frac{1}{\mu}\left(\frac{1}{K}\left(\frac{\lfloor \x_1 M\rfloor+R}{M|\N_1|}\right)^{1-\beta}+\frac{1}{K}\sum_{k=2}^K\left(\frac{\lfloor \x_k M\rfloor}{M|\N_k|}\right)^{1-\beta}\right)^{\frac{1}{1-\beta}},\]
where $R=M-\sum_{k\in[K]} \lfloor \x_k M\rfloor$ and $\mu =\frac{1}{K}\left(\frac{\lfloor \x_1 M\rfloor+R}{M|\N_1|} +\sum_{k=2}^K\frac{\lfloor \x_k M\rfloor}{M|\N_k|}\right)$.
\end{theorem}

The bound in Theorem~\ref{thm:worstinequality} captures the worst-case inequality arising solely from item indivisibility. Its apparent looseness is primarily due to \LL{numerical} rounding, \ie Lemma~\ref{lem:similarmargin} shows that allocation of each group deviates from its ideal fractional share by at most $O(K)$, yielding a relative error of order $K/|\M_\ell|$. Consequently, when item partitions are sufficiently large compared to the number of groups, the induced Atkinson inequality is small, and the bound is effectively tight. Although the bound is formally pessimistic, it cannot approach $1$ under typical settings. In particular, doing so would require extremely small partitions, many groups, and highly imbalanced group sizes, which are atypical in practical applications. \textit{Thus, Theorem~\ref{thm:worstinequality} should be viewed as a robustness guarantee such that indivisibility introduces a bounded fairness loss that vanishes as partition sizes grow}.

\section{Optimality for General Inequality Metrics and Supplementary  Constraints}\label{sec:generalmetric} 

\subsection{General Inequality Metrics}\label{sec:generalInequality}

Besides the Atkinson Inequality, there are other common inequality metrics used to quantify fairness in the literature, including {\em statistical parity}~\cite{HertweckHL21}, {\em gini index}~\cite{Farris10,charles2022gini}, {\em nash welfare}~\cite{kaneko1979nash,CaragiannisKMPS19}, and {\em counterfactual fairness}~\cite{KusnerLRS17}. The continuous version \oursc is applicable to a broader range of inequality metrics. That said, given any general inequality metric, there always exists a setting of $\alpha\in (0,1)^K$ such that the corresponding output of \oursc yields the minimum inequality w.r.t. the given inequality metric. Formally, we have the following result. 

\begin{corollary}\label{cor:generalmetrics}
Let $\I(\x)$ be an inequality metric and $\x^\ast := \arg \min I(\x)$ s.t.  $\sum^K_{k=1} \x_k=1, \x \in (0,1)^K$ exist. There exists a parameter setting $\alpha^\ast = (\alpha_1, \alpha_2, \ldots, \alpha_K) \in (0,1)^K$ s.t.  $\x^\ast$ is the output of \oursc.
\end{corollary}

Corollary~\ref{cor:generalmetrics} directly follows from  Theorem~\ref{thm:generalalpha} and reveals that \oursc can acquire the optimal solution across a range of inequality metrics $\I$. This result demonstrates that our framework \ours is applicable to a broad class of inequality metrics, with its optimal solution approximating the minimum inequality. Since the Atkinson inequality and Nash Welfare are widely adopted metrics in resource allocation applications, our study highlights these two measures. \vspace{-1mm}

\subsection{Welfare Maximization under Supplementary  Constraints}\label{sec:ExtCon}

\begin{algorithm}[!t]
\begin{small}
\caption{Adaptive Resource Allocation \oursa}\label{alg:RAC}
\KwIn{Agent set $\N$, Item set $\M$, inequality-aversion parameter $\alpha$, group number $K$, partition number $L$, supplementary constraints $\C=\{\C_1, \C_2, \cdots, \C_L\}$}
\KwOut{Utility matrix $\U$}
Initialize $\U\gets \{0\}^{K\times L}$\;
\For{$\ell \gets 1$ to $L$}
{
  \For{$t \gets 1$ to $|\M_\ell|$}
  {
    $\textrm{gain}[k] \gets \big(\tfrac{U[k,\ell]+1}{|M_\ell|}\big)^{\alpha_k}-\big(\tfrac{U[k,\ell]}{|M_\ell|}\big)^{\alpha_k}\ \textrm{for } k\in [K]$\;
    \While{ $\sum_{k\in[K]}\textrm{gain}[k] > 0$}
    {
     $\S^\ast \gets \arg\max_{k\in [K]}\textrm{gain}[k]$\;
     $k^\ast \sim \mathrm{Uniform}(\S^\ast)$\;
     $\textrm{gain}[k^\ast] \gets 0$\;
     \If{$\U[k^\ast,\ell]+1$ does not violate constraint $\C_\ell$}
     {
        $\U[k^\ast,\ell]\gets \U[k^\ast,\ell]+1$\;
        \KwBreak\;
     }
    }
  }
}
\Return $\U$\; 
\end{small}
\end{algorithm}

In real-world applications, there are usually supplementary application-specific box constraints to consider. For example, in university course assignments across multiple departments, each course is subject to a predefined maximum capacity, and courses need a minimum enrollment to be offered. Meanwhile,  students submit unique course selection requests for various courses. By considering students within the same department as a group, the objective is to maximize group fairness in course allocation across departments, while adhering to the course capacity constraints. In this scenario, \ours with a minor adaptation is still able to achieve the best possible group fairness while adhering to the supplementary constraints. We present the pseudo-code of \ours with adaptation, termed \oursa, in Algorithm~\ref{alg:RAC}.

Let $\C=\{\C_1, \C_2, \cdots, \C_L\}$ be the supplementary constraints for the $L$ item partitions, respectively. In general, each constraint $\C_\ell$ governs the allocation within the $\ell$-th partition with a group-wise separable upper bounds for $\ell \in [L]$. Similar to Algorithm~\ref{alg:RA}, when allocation is on partition $\M_\ell$, the index $k^\ast$ of the group with the largest marginal gain is identified before assignment. Subsequently, we check whether \LL{assigning an item from partition $\M_\ell$ to an agent in group $\N_{k^\ast}$} breaches the constraint $\C_\ell$. This process is repeated until we locate a group where the largest marginal gain can be achieved without violating the constraint $\C_\ell$. We demonstrate that with this minor  modification, \oursa is able to attain an allocation with the  maximum welfare, subject to the additional constraint $\C$. \LL{A subtle point is that unlike in inequality, there is no optimality gap in the maximum social welfare achieved by \ours and \oursa.}  

\begin{theorem}(Optimality of \oursa under Box Constraints)\label{thm:oursa}
Consider the allocation of item set $\M$ with $L$ partitions to agent set $\N$ of multiple groups, subject to diversity constraints and a set of per-partition constraints $\C = \{\C_1, \ldots, \C_L\}$. For each partition $\M_\ell$, constraint $\C_\ell$ is group-wise separable and is of the form $0 \le \U[k,\ell] \le c_{k,\ell}$ for some $c_{k,\ell} \in \mathbb{N}_+$. The allocation returned by \oursa maximizes the welfare for the given $\alpha$ configuration under the constraints $\C$. 
\end{theorem}\vspace{-1mm}

\section{Related Work}\label{sec:related}

\spara{Fair Resource Allocation under Diversity Constraints} Fair resource allocation under diversity constraints has been extensively studied across multiple domains, especially in the context of indivisible goods and group fairness. A prominent line of work focuses on incorporating diversity constraints into public housing allocation. \citet{BenabbouCHSZ18,BenabbouCZ19} investigate the Singapore housing system, where houses distributed across blocks are assigned to agents from different ethnic groups, with strict upper bounds on group-wise occupancy. They design a $\tfrac{1}{2}$-approximation algorithm for this constrained allocation problem. \citet{AGSW19} further reduce diversity constraints to regional quotas in polynomial time, unifying both formulations. Beyond additive valuations, \citet{BenabbouCIZ20} generalize the objective to matroid rank functions, allowing more expressive value systems under diversity constraints. While their solution achieves Nash social welfare maximization and envy-freeness up to one good (EF1), their focus remains on group-constrained allocation through submodular optimization. Inspired by these works, \citet{GrossHumbertBB21} incorporate both agents’ preferences and neighborhood similarity into a generic utility function, and explore sequential allocation and swap-based mechanisms. They find that the former lacks swap-stability, while the latter may reduce social welfare. \citet{BanerjeeE023} extend the setting by introducing priority-respecting allocations: each group has a quota and a priority-ordered list of eligible agents, and the goal is to ensure Pareto efficiency under quota, eligibility, and priority constraints. 

\spara{Relaxed Fairness for Indivisible Goods Allocation} While EF1 has been considered in group-based settings (e.g., \cite{BenabbouCIZ20}), another line of research focuses on relaxed fairness notions for individual agents in the allocation of indivisible goods. \citet{ProcacciaW14, KurokawaPW16, KurokawaPW18} initiate the study of maximin share (MMS) guarantees, proving that a $\tfrac{2}{3}$-approximation can be ensured relative to optimal divisible allocations. \citet{BarmanKV18} relax envy-freeness to EF1 and show its feasibility in polynomial time, further developing a pseudo-polynomial time algorithm to achieve both EF1 and Pareto efficiency. \citet{SegalHaleviS18} propose democratic fairness, where fairness is satisfied for a fraction of agents in each group. They show that a $\tfrac{1}{2}$ fraction is optimal under EF1 constraints. \citet{BenabbouCEZ19} extend this line to typewise fairness, introducing the notion of waste and demonstrating that maximizing marginal utility yields typewise EF1. \citet{GrossHumbertBB23} critique existing group envy-freeness notions and propose a new metric quantifying group envy, which can be approximated via sampling. \citet{ScarlettTZ23} analyze the simultaneous satisfaction of individual envy-freeness (i-EF) and group-weighted envy-freeness (g-WEF), showing polynomial-time algorithms under three valuation settings: (i) identical additive valuations across agents, (ii) group-shared valuations, and (iii) heterogeneous valuations.

\spara{Online Fair Allocation under Diversity Constraints} Recent work also addresses online settings where resources arrive sequentially and must be allocated in real time under diversity constraints. \citet{XX22} study the online allocation of goods by non-profit platforms, where each agent belongs to one or more groups and the system aims to ensure fair group-wise shares proportional to predefined ratios. They propose two sampling-based algorithms using linear programming formulations. \citet{LouisNNS23} consider online bipartite matching problems under two key constraints: proportional fairness, where assigned item proportions must lie within bounds; and diversity, where minimum group-wise allocation thresholds must be met. They provide approximation algorithms for these combined objectives. \citet{BeiLPW20} studies candidate selection under proportional fairness constraints and gives polynomial-time algorithms for finding the largest feasible subset. In the context of social commerce platforms, \citet{GuptaNCVRNDC23} explore exposure allocation from producers to resellers through social networks. The problem is modeled with two-sided cardinality constraints, ensuring each product is assigned to a minimum number of resellers and vice versa. A mixed-integer programming approach is proposed, approximating Nash social welfare in near-optimal ways. 

Existing methods rely on hard constraints to enforce diversity, which limits the flexibility of the optimization process and often results in suboptimal fairness. In contrast, our method \ours adopts an inequality-aversion parameterized mechanism, enabling it to achieve minimal inequality across a broad class of inequality metrics. Our results also extend to the case where there are application-specific supplementary constraints. 
\section{Experiments}\label{sec:exp}

\subsection{Experimental Settings}\label{sec:expsetting}

\spara{Datasets} We consider three real-world applications, \ie Singapore HDB allocation, course assignment in universities, and job recommendation. The details of the datasets are in Appendix~\ref{sec:extraExp}.

\begin{table}[!t]
\centering
\caption{Department-level statistics of course applications} \label{tbl:course}
\setlength{\tabcolsep}{0.5em}
\small
\vspace{-2mm}
\resizebox{0.48\textwidth}{!}{%
\begin{tabular}{@{}l|ccc@{}}
\toprule
{\bf Department} & {\bf \#Students} & {\bf \#Humanities Applications} & {\bf \#Math Applications} \\ \midrule
D01 & 25  & 81   & 24  \\
D02 & 61  & 75   & 70  \\
D03 & 65  & 198  & 31  \\
D04 & 290 & 1253 & 103 \\
D05 & 178 & 637  & 141 \\
D06 & 43  & 126  & 22  \\
D07 & 211 & 714  & 133 \\
D08 & 70  & 297  & 29  \\
D09 & 160 & 571  & 74  \\
D10 & 528 & 2547 & 544 \\
D11 & 80  & 326  & 28  \\
D12 & 268 & 1244 & 205 \\
D13 & 108 & 317  & 120 \\ \midrule
{\bf Total} & {\bf 2087} & {\bf 8386} & {\bf 1524} \\ \bottomrule
\end{tabular}}\vspace{-1mm}
\end{table}

\begin{figure*}[!t]
\centering
\begin{minipage}{0.31\linewidth}
\centering
\begin{small}
\begin{tikzpicture}
\begin{axis}[
    height=\columnwidth/1.3,
    width=\columnwidth/1,
    enlarge x limits=true,
    ymin=0, ymax=80000,
    xmin=0.8, xmax=5.2,
    xlabel={$\beta$ in {\em Atkinson Inequality}},
    ylabel={\em Inequality scores},
    xtick={{1},{2},{3},{4},{5}},
    xticklabels={0.1,0.2,0.4,0.5,0.8},
    xticklabel style={font=\small},
    ytick={10, 100, 1000, 10000},
    yticklabels={$10^{-4}$, $10^{-3}$,$10^{-2}$, $10^{-1}$},
    yticklabel style={font=\small},
    ymode=log,
    legend columns=3,
    legend style={fill=none,font=\footnotesize,at={(0.5,1.1)},anchor=center,draw=none},
]
    \addplot[line width=0.25mm,mark=o,color=red] coordinates {(1,9.45622384)(2,18.916)(3,37.843)(4,47.312)(5,75.735)};
    \addplot[line width=0.25mm,mark=diamond,color=cyan] coordinates {(1, 948.111)(2,1496.095)(3,4278.891)(4,7931.542)(5,11512.855)};
    \addplot[line width=0.25mm,mark=diamond,color=blue]  coordinates {(1, 4300.944)(2,8779.034)(3,18272.713)(4,23288.019)(5,39101.064)};
    \legend{{\ours},{\ATC},{\ProFair}}
\end{axis}
\end{tikzpicture}
\end{small}\vspace{-1mm}
\captionof{figure}{Fairness on Singapore HDB Allocation.} \label{fig:fairnessatk}
\end{minipage}\hspace{1mm}
\begin{minipage}{0.31\linewidth}
\centering
\begin{small}
\begin{tikzpicture}
\begin{axis}[
    height=\columnwidth/1.3,
    width=\columnwidth/1,
    enlarge x limits=true,
    ymin=10, ymax=1000,
    xmin=0.8, xmax=5.2,
    xlabel={$\beta$ in {\em Atkinson Inequality}},
    ylabel={\em Inequality scores},
    xtick={{1},{2},{3},{4},{5}},
    xticklabels={0.1,0.2,0.4,0.5,0.8},
    xticklabel style={font=\small},
    ytick={10, 100, 1000},
    yticklabels={$10^{-4}$, $10^{-3}$,$10^{-2}$},
    yticklabel style={font=\small},
    ymode=log,
    legend columns=3,
    legend style={fill=none,font=\footnotesize,at={(0.5,1.1)},anchor=center,draw=none},
]
    \addplot[line width=0.25mm,mark=o,color=red] coordinates {(1, 30.8609) (2, 61.5926) (3, 122.6575) (4, 152.9849) (5, 243.11)};
    \addplot[line width=0.25mm,mark=diamond,color=cyan] coordinates {(1, 60.35)(2,118.63)(3,233.189)(4,312.1416)(5,430.7666)};
    \addplot[line width=0.25mm,mark=diamond,color=blue]  coordinates {(1, 75.072550) (2, 199.116530) (3, 491.472530) (4, 611.666110) (5, 753.880170)};
    \legend{{\oursa},{\ATC},{\ProFair}}
\end{axis}
\end{tikzpicture}
\end{small}\vspace{-1mm}
\captionof{figure}{Fairness on Course Assignment.} \label{fig:coursefairness}
\end{minipage}\hspace{1mm}
\begin{minipage}{0.31\linewidth}
\centering
\begin{small}
\begin{tikzpicture}
\begin{axis}[
    height=\columnwidth/1.3,
    width=\columnwidth/1,
    enlarge x limits=true,
    ymin=0, ymax=1000,
    xmin=0.8, xmax=5.2,
    xlabel={$\beta$ in {\em Atkinson Inequality}},
    ylabel={\em Inequality scores},
    xtick={{1},{2},{3},{4},{5}},
    xticklabels={0.1,0.2,0.4,0.5,0.8},
    xticklabel style={font=\small},
    ytick={10, 100, 1000},
    yticklabels={$10^{-4}$,$10^{-3}$,$10^{-2}$},
    yticklabel style={font=\small},
    ymode=log,
    legend columns=2,
    legend style={fill=none,font=\footnotesize,at={(0.45,1.1)},anchor=center,draw=none},
]
    \addplot[line width=0.25mm,mark=o,color=red] coordinates {(1, 8.73) (2, 17.457) (3, 34.915) (4, 43.645) (5, 69.833)};
    \addplot[line width=0.25mm,mark=diamond,color=blue] coordinates {(1,30.418)(2,60.843)(3,121.709)(4,152.145)(5,243.447)};
    \legend{{\oursa},{\ProFair}}
\end{axis}
\end{tikzpicture}
\end{small}\vspace{-1mm}
\captionof{figure}{Fairness on Job Recommendation.} \label{fig:jobfairness}
\end{minipage}
\end{figure*}

\begin{figure*}[!t]
\centering
\begin{minipage}{0.31\linewidth}
\centering
\begin{small}
\begin{tikzpicture}
\begin{axis}[
    height=\columnwidth/1.3,
    width=\columnwidth/1,
    enlarge x limits=true,
    ymin=0, ymax=1.1,
    xmin=0.8, xmax=5.2,
    xlabel={Agent number $|\N|$},
    ylabel={Nash Welfare},
    xtick={{1},{2},{3},{4},{5}},
    xticklabels={1000, 2000, 4000, 5000, 8000},
    xticklabel style={font=\small},
    yticklabel style={font=\small},
    legend columns=3,
    legend style={fill=none,font=\footnotesize,at={(0.5,1.1)},anchor=center,draw=none},
]
    \addplot[line width=0.25mm,mark=o,color=red] coordinates {(1,0.99094144)(2,0.49547072)(3,0.24773536)(4,0.19818829)(5,0.12386768)};
    \addplot[line width=0.25mm,mark=diamond,color=cyan] coordinates {(1,0.98661418)(2,0.4920434)(3,0.24587645)(4,0.19676925)(5,0.12293423)};
    \addplot[line width=0.25mm,mark=diamond,color=blue]  coordinates {(1,0.86663796)(2,0.43331898)(3,0.21665949)(4,0.17332759)(5,0.10832975)};
    \legend{{\ours},{\ATC},{\ProFair}}
\end{axis}
\end{tikzpicture}
\end{small}\vspace{-1mm}
\captionof{figure}{Nash Welfare on Singapore HDB Allocation.} \label{fig:nashwelfare}
\end{minipage}\hspace{1mm}
\begin{minipage}{0.31\linewidth}
\centering
\begin{small}
\begin{tikzpicture}
\begin{axis}[
    height=\columnwidth/1.3,
    width=\columnwidth/1,
    enlarge x limits=true,
    ymin=0, ymax=1.1,
    xmin=0.8, xmax=6.2,
    xlabel={Course quota $\tau$},
    ylabel={Nash Welfare},
    xtick={{1},{2},{3},{4},{5}, {6}},
    xticklabels={0.05, 0.1, 0.15, 0.2, 0.25, 0.3},
    xticklabel style={font=\small},
    yticklabel style={font=\small},
    legend columns=3,
    legend style={fill=none,font=\footnotesize,at={(0.5,1.1)},anchor=center,draw=none},
]
    \addplot[line width=0.25mm,mark=o,color=red] coordinates {(1, 0.39448832) (2, 0.77793705) (3, 0.95114229) (4, 1.04108268) (5, 1.06186439) (6, 1.07103885)};
    \addplot[line width=0.25mm,mark=diamond,color=cyan] coordinates {(1, 0.38671223) (2, 0.70263019) (3, 0.84097777) (4, 0.92946393) (5, 0.9731392) (6, 0.9884295)};
    \addplot[line width=0.25mm,mark=diamond,color=blue]  coordinates {(1, 0.356649685) (2, 0.581459662) (3, 0.6838966) (4, 0.787641715) (5, 0.781294309) (6, 0.797594724)};
    \legend{{\oursa},{\ATC},{\ProFair}}
\end{axis}
\end{tikzpicture}
\end{small}\vspace{-1mm}
\captionof{figure}{Nash Welfare on Course Assignment.} \label{fig:nashwelfarecourse}
\end{minipage}\hspace{1mm}
\begin{minipage}{0.31\linewidth}
\centering
\begin{small}
\begin{tikzpicture}
\begin{axis}[
    height=\columnwidth/1.3,
    width=\columnwidth/1,
    enlarge x limits=true,
    ymin=0, ymax=3.8,
    xmin=0.8, xmax=6.2,
    xlabel={Job quota $\tau$},
    ylabel={Nash Welfare},
    xtick={{1},{2},{3},{4},{5}, {6}},
    xticklabels={0.05, 0.1, 0.15, 0.2, 0.25, 0.3},
    xticklabel style={font=\small},
    yticklabel style={font=\small},
    legend columns=2,
    legend style={fill=none,font=\footnotesize,at={(0.45,1.1)},anchor=center,draw=none},
]
    \addplot[line width=0.25mm,mark=o,color=red] coordinates {(1, 1.34619028) (2, 2.12835937) (3, 2.67165029) (4, 3.07325984) (5, 3.39537625) (6, 3.66379047)};
    \addplot[line width=0.25mm,mark=diamond,color=blue] coordinates {(1, 0.35597966) (2, 0.45076613) (3, 0.53665463) (4, 0.61475253) (5, 0.68153246) (6, 0.741721766)};
    \legend{{\oursa},{\ProFair}}
\end{axis}
\end{tikzpicture}
\end{small}\vspace{-1mm}
\captionof{figure}{Nash Welfare on Job Recommendation.} \label{fig:nashwelfarejob}
\end{minipage}
\end{figure*}

\begin{figure*}[!t]
  \centering
  \begin{minipage}[t]{0.31\textwidth}
    \centering
    \begin{tikzpicture}[scale=1,every mark/.append style={mark size=1.5pt}]
        \begin{axis}[
            height=\columnwidth/1.3,
            width=\columnwidth/1,
            xlabel={$\beta$ in {\em Atkinson Inequality}},
            ylabel={\em Inequality scores},
            xmin=0.5, xmax=5.5,
            ymin=0, ymax=500000,
            xtick={1,2,3,4,5},
            xticklabels={0.1,0.2,0.4,0.5,0.8},
            ytick={10,1000,100000},
            yticklabels={$10^{-4}$,$10^{-2}$,$1$},
            xticklabel style={font=\scriptsize},
            yticklabel style={font=\footnotesize},            
            ymode=log,
            legend columns=2,
            legend style={fill=none,font=\footnotesize,at={(0.5,1.15)},anchor=center,draw=none},
        ]
        \addplot[line width=0.25mm,mark=o,color=red]
        coordinates {
            (1, 9.45622384) (2, 18.916) (3, 37.843) (4, 47.312) (5, 75.735)};      
        \addplot[line width=0.25mm,mark=diamond,color=blue]
        coordinates {
            (1, 30227.8269) (2, 62612.9978) (3,113514.4062) (4, 151471.0099) (5, 297517.1859)};       
        \legend{\ours, \ours$_{\textrm{rand}}$}
        \end{axis}
    \end{tikzpicture}\vspace{-1mm}
    \caption{Impacts of $\alpha$ on Fairness.}
    \label{fig:alphafairness}
  \end{minipage}\hspace{1mm}
  \begin{minipage}[t]{0.31\textwidth}
    \centering
    \begin{tikzpicture}[scale=1,every mark/.append style={mark size=1.5pt}]
        \begin{axis}[
            height=\columnwidth/1.3,
            width=\columnwidth/1,
            xlabel={Agent number $|\N|$},
            ylabel={Nash Welfare},
            xmin=0.5, xmax=5.5,
            ymin=0, ymax=1.1,
            xtick={1,2,3,4,5},
            xticklabels={1000, 2000, 4000, 5000, 8000},
            xticklabel style={font=\scriptsize},
            yticklabel style={font=\footnotesize},
            legend columns=2,
            legend style={fill=none,font=\footnotesize,at={(0.5,1.15)},anchor=center,draw=none},
        ]
        \addplot[line width=0.25mm,mark=o,color=red]
        coordinates {
            (1, 0.99094) (2, 0.49547) (3, 0.24774) (4, 0.19819) (5, 0.12387)};       
        \addplot[line width=0.25mm,mark=diamond,color=blue]
        coordinates {
            (1, 0.12435) (2, 0.12388) (3,0.10940) (4, 0.08901) (5, 0.06274)};        
        \legend{\ours, \ours$_{\textrm{rand}}$}
        \end{axis}
    \end{tikzpicture}\vspace{-1mm}
    \caption{Impacts of $\alpha$ on Nash Welfare.}
    \label{fig:alphawelfare}
  \end{minipage}\hspace{1mm}
  \begin{minipage}[t]{0.31\textwidth}
    \centering
    \begin{tikzpicture}[scale=1,every mark/.append style={mark size=1.5pt}]
        \begin{axis}[
            height=\columnwidth/1.3,
            width=\columnwidth/1,
            xlabel={Parameter $\alpha_1$},            
            xmin=0.5, xmax=5.5,
            ymin=0, ymax=0.6,
            xtick={1,2,3,4,5},
            xticklabels={0.01,0.1,0.2,0.4,0.8},
            xticklabel style={font=\scriptsize},
            yticklabel style={font=\footnotesize},
            legend columns=2,
            legend style={fill=none,font=\footnotesize,at={(0.5,1.15)},anchor=center,draw=none},
        ]
        \addplot[line width=0.25mm,mark=o,color=red]
        coordinates {(1, 0.45428) (2, 0.47312) (3, 0.47312) (4, 0.47312) (5, 0.47312)};      
        \addplot[line width=0.25mm,mark=diamond,color=blue]
        coordinates {(1, 0.49547) (2, 0.49547) (3,0.49547) (4, 0.49547) (5, 0.49547)};        
        \legend{Ineq. ($\times 10^{-3}$), Welfare}
        \end{axis}
    \end{tikzpicture}\vspace{-1mm}
    \caption{Sensitivity of $\ours$ on $\alpha_1$.}
    \label{fig:alpha1}
  \end{minipage}
\end{figure*}

\spara{Baselines} In the research of group fairness under diversity constraints, two existing methods \ATC~\cite{BenabbouCHSZ18} and \ProFair~\cite{LouisNNS23} are closely related. We compare \ours against them for the three applications under the fairness and diversity constraints. Distinct from \ours, both \ATC and \ProFair utilize explicit hard constraints for each group of agents in the optimization problem.

\spara{Settings} In the Singapore HDB allocation, the set of flats (items) from five partitions follows Fig.~\ref{fig:HDBdata} as $\M=\{218, 114, 211, 327, 120\}$ in a total of $990$ houses. We consider four groups of agents, \ie {\em Chinese}, {\em Malays}, {\em Indians}, and {\em Others}. We fix their ratios as $\{74.0\%, 13.5\%, 9.0\%, 3.4\%\}$ and consider an agent set with $|\N|=2000$. We set parameters $\alpha_1 = 0.1$ and $\epsilon = 0.00001$ for \ours as default. By following the settings in~\cite{BenabbouCHSZ18,LouisNNS23}, we utilize random utility for \ATC and adopt uniform utility for \ours and \ProFair. 

In the assessment, we adopt metrics {\em Atkinson inequality} and {\em Nash welfare} (NW)~\cite{BenabbouCIZ20,Kell023, GuptaNCVRNDC23} to quantify the fairness and welfare among groups, respectively. For Atkinson inequality in Equation~\eqref{eqn:atkinson}, we consider $\beta\in\{0.1,0.2,0.4,0.5,0.8\}$. For HDB allocation, 
we vary agent number $|\N|=\{1000, 2000, 4000, 5000, 8000\}$ while keeping the number of flats constant at $990$. For course assignment and job recommendation, we vary the corresponding resource quota $\tau \in \{0.05, 0.1, 0.15, 0.2, 0.25, 0.3\}$ to measure the robustness of tested methods in terms of Nash Welfare. 

\subsection{Fairness \& Welfare in Singapore HDB Allocation}\label{sec:HDBresults}

Fig.~\ref{fig:fairnessatk} and Fig.~\ref{fig:nashwelfare} display the inequality scores and Nash welfare, respectively, of the three tested methods on the Singapore HDB dataset. For the benefit of illustration, we scale the scores by a factor of $10^5$. The inequality is measured using the Atkinson inequality, parameterized by $\beta\in (0,1)$, which modulates the sensitivity to distribution variances. \eat{As $\beta$ increases, the index becomes more responsive to variations in inequality.} We examine the results across a spectrum of $\beta$ values. We observe that \ours achieves significantly smaller inequality scores than those of the other two methods across all tested $\beta$ values. In particular, the inequality value achieved by \ours ranges  from $9.46 \times 10^{-5}$ to $75.74 \times 10^{-5}$, approximately \LL{$0.6\% - 1.0\%$} and $0.2\%$ of the scores obtained  by \ATC and \ProFair, respectively. In Theorem~\ref{thm:generalalpha}, we prove that by judiciously choosing the parameter $\alpha$, our model \ours can ideally attain the minimal inequality value of zero. The non-zero scores observed in \ours result from integer numerical errors, as proved in Theorem~\ref{thm:worstinequality}. In addition, inequality scores rise along with the increase of $\beta$ as expected. This is consistent with the understanding that larger $\beta$ values exhibit higher sensitivity to inequality.

In Fig.~\ref{fig:nashwelfare}, \ours achieves the highest Nash welfare, exhibiting a marginal benefit compared to \ATC and a significant advantage over \ProFair. Specifically, for agent number $|\N|=1000$, \ours achieves Nash welfare of $0.991$ while the corresponding values of \ATC and \ProFair are $0.987$ and $0.867$ respectively. Similar scenarios are observed in the remaining cases. Meanwhile, Nash welfare diminishes as the number of agents grows while the available resources (flats) remain constant, attributable to a decreased share per agent.

\subsection{Fairness \& Welfare in Course Assignment}\label{sec:courseresults}

Fig.~\ref{fig:coursefairness} and Fig.~\ref{fig:nashwelfarecourse} respectively illustrate the inequality scores and Nash welfare achieved by the three methods in the course assignment task. As shown, the results in Fig.~\ref{fig:coursefairness} follow a similar pattern in Fig.~\ref{fig:fairnessatk}. Specifically, \oursa achieves the lowest inequality scores compared with the other two baselines. In particular, the Atkinson inequality scores by \oursa are around $50\% - 56\%$ and $24\% - 40\%$ of those scores by \ATC and \ProFair, respectively, ranging from $30.86\times 10^{-5}$ to $243.11\times 10^{-5}$. As discussed in Section~\ref{sec:HDBresults}, these scores exhibit a rising trend with increasing values of the parameter~$\beta$.

In Fig.~\ref{fig:nashwelfarecourse}, we observe that \oursa consistently achieves the highest Nash welfare values among the three methods with a clear advantage across the variations of quotas. When the quota $\tau=0.15$, the Nash welfare of \oursa is $1.13\times$ and $1.39\times$ of the Nash welfare by \ATC and \ProFair, respectively. Furthermore, Table~\ref{tbl:courseallocate} reports the number of courses successfully allocated under the given constraints by the three evaluated methods. As evidenced by the results, \oursa achieves the highest allocation count, demonstrating its superior effectiveness.

\subsection{Fairness \& Welfare in Job Recommendation}\label{sec:jobresults}

\begin{table}[!t]
\centering
\caption{Total number of allocated courses} \label{tbl:courseallocate}
\setlength{\tabcolsep}{0.5em}
\small
\vspace{-2mm}
\begin{tabular}{@{}l|cc@{}}
\toprule
{\bf Methods} & {\bf \#Humanities} & {\bf \#Mathematics}  \\ \midrule
\ours & 1484  & 766    \\
\ATC & 1195  & 565  \\
\ProFair & 1398  & 689 \\ \bottomrule
\end{tabular}
\end{table}

The Fairness and Nash welfare outcomes for the job recommendation task are shown in Fig.~\ref{fig:jobfairness} and Fig.~\ref{fig:nashwelfarejob}, respectively. Due to an out-of-memory (OOM) error on our server (16GB RAM), the baseline \ATC could not be evaluated, and thus only the results of \oursa and \ProFair are reported. As illustrated, \oursa consistently achieves substantially lower inequality scores and higher Nash welfare compared to \ProFair. In particular, the inequality score is reduced by up to $28\%$, while the Nash welfare improves by as much as $5\times$, aligning with the trends observed in the previous two applications. The findings provide strong evidence for the superior performance of \oursa. 

\subsection{Ablation Study}\label{sec:ablation}

\spara{Impacts of $\alpha$} To evaluate the crucial role of $\alpha$ in \ours, we develop a variant of \ours by adopting randomized $\alpha$ for comparison, termed as \ours$_{\textrm{rand}}$. Specifically, we randomize $\alpha$ in $(0,1)$ and report the average performance of \ours$_{\textrm{rand}}$ through $10$ runs. We compare \ours with \ours$_{\textrm{rand}}$ in terms of both inequality scores and Nash welfare in Figure~\ref{fig:alphafairness} and Figure~\ref{fig:alphawelfare}, respectively. Similar to Figure~\ref{fig:fairnessatk}, the values have been scaled by a factor of $10^5$ for better demonstration.

As shown in Figure~\ref{fig:alphafairness}, inequality scores of \ours$_{\textrm{rand}}$ are $10^4$ times larger than those of \ours. In the scenario where $\beta=0.8$ with the highest sensitivity to inequality, \ours achieves the inequality score of $75.735 \times 10^{-5}$. In contrast, the score of \ours$_{\textrm{rand}}$ attains a substantially higher score of $297517 \times 10^{-5}$, exhibiting a significant four-order-of-magnitude disparity between the two. Nash welfare is shown in Figure~\ref{fig:alphawelfare}. When agent number $|\N|$ increases from $1000$ to $8000$, welfare by \ours drops from $0.991$ to $0.124$ while the value by \ours$_{\textrm{rand}}$ maintains around $0.1$. This shows that random $\alpha$ settings can lead to arbitrarily worse allocation. These observations clearly show the crucial role of $\alpha$ in \ours.

\spara{Robustness of \ours in $\alpha_1$} In Algorithm~\ref{alg:alphas}, we set $\alpha_1=0.1$ by default in the calculation of all remaining $\alpha$ values. To examine the sensitivity of \ours towards the initial value of $\alpha_1$, we initialize $\alpha_1 \in \{0.01,0.1,0.2,0.4,0.8\}$ and then report the corresponding fairness and Nash welfare in Figure~\ref{fig:alpha1}. As shown, both the inequality scores and Nash Welfare remain relatively stable despite the dramatic increase of the initial $\alpha_1$. This finding substantiates the robust performance of \ours across varying values of $\alpha_0$. Specifically, for any $\alpha_1 \in (0,1)$, Algorithm~\ref{alg:alphas} consistently identifies an $\alpha$ setting that enables \ours to attain near-optimal performance.

\vspace{-2mm}
\section{Conclusion}\label{sec:conclusion}

In this paper, we proposed \ours, a parameterized framework for resource allocation under diversity constraints. \ours achieves welfare-optimal allocations while flexibly controlling inequality across groups. We established fairness optimality for a broad class of inequality metrics. The inequality gap of indivisible items from the divisible optimum is also established. Furthermore, we extend the framework to an adaptive version \oursa that accommodates supplementary application-specific constraints while preserving welfare optimality. Extensive experiments on real-world applications validate the effectiveness and robustness of our approach.

\sketch{
\section{Acknowledgments}
Huang's research was supported by Fundamental and Interdisciplinary Disciplines Breakthrough Plan of the Ministry of Education of China (No. JYB2025XDXM118), National Natural Science Foundation of China (No. U25B2023), and Hubei Provincial Natural Science Foundation of China (No. 2026AFA002). Xiao's research was supported by the Ministry of Education, Singapore, under its MOE AcRF TIER 3 Grant (MOE-MOET32022-0001). Lakshmanan’s research was supported by an NSERC grant (RGPIN2020-05408) and an NSERC Alliance Grant (ALLRP 599949 – 24).
}

\clearpage
\bibliographystyle{ACMRefer}
\balance
\bibliography{ref}

\clearpage
\appendix
\section{Appendix}\label{sec:app}

\subsection{Experimental Datasets}\label{sec:extraExp}

\noindent
{\bf (i) Singapore HDB Dataset}.  As shown in Figure~\ref{fig:HDBdata}, the available HDB flats are located in five regions in total of $990$ flats, namely West region with $218$ flats, North region with $114$ flats, North-East region with $211$ flats, Central region with $327$ flats, and East region with $120$ flats. Meanwhile, according to the Department of Statistics of Singapore in 2023~\cite{Population23}, there are $74.0\%$ {\em Chinese}, $13.5\%$ {\em Malays}, $9.0\%$ {\em Indians}, and $3.4\%$ {\em others} among the population. \\ {\bf (ii) Course Assignment Dataset}~\cite{LouisNNS23}. This dataset captures a university course assignment setting involving students from $13$ departments where students are required to enroll in two categories of courses: mathematics and humanities. Within each category, there are multiple individual courses, each subject to a predefined capacity constraint. In total, mathematics courses offer $931$ quotas, while humanities courses offer $3197$ quotas. To ensure equitable access for each course, no more than $\tau=30\%$ of its capacity can be allocated to students from any single department. The dataset statistics are in Table~\ref{tbl:course}. {\bf (iii) Job Recommendation Dataset}~\cite{VladimirovaPD24}. \arxiv{This dataset is collected for fairness investigation in job recommendation. It contains $30361$ job applicants and $57355$ job advertisements. The dataset exhibits extreme class imbalance with an overall click-through rate of only $0.7\%$. In particular, each applicant has $41$ job-hunting related features (6 categorical and 35 numerical attributes) and one protected attribute such as gender, and each job has $7$ categorical attributes for description. The label is a binary value, indicating if job advertisements are clicked by job applicants. For our purpose, we preprocess this dataset as follows. First, we perform cold-start data splitting by dividing job applicants and jobs into disjoint sets with a 60/20/20 ratio for training, validation, and testing. Second, we apply label encoding to categorical features and standard scaling to numerical features. Third, to address the severe class imbalance, we apply negative sampling with a 5:1 ratio during training while retaining all samples for evaluation. Finally, we train a two-tower model that generates job applicants and jobs embeddings through separate neural networks, computing the relevant score between a job applicant and a job. Eventually, the generated matrix records the protected attributes of $6073$ job applicants and their relevant scores on $11471$ advertised jobs.}\sketch{This dataset is collected for fairness investigation in job recommendation, containing $30{,}361$ job applicants and $57{,}355$ job advertisements with an overall click-through rate of $0.7\%$. Each applicant has $41$ features (6 categorical and 35 numerical) and one protected attribute (e.g., gender), while each job contains $7$ categorical attributes. The binary label indicates whether a job advertisement is clicked by a job applicant. For preprocessing, we perform cold-start splitting with a 60/20/20 ratio for training, validation, and testing, apply label encoding and feature standardization, and use 5:1 negative sampling during training to alleviate class imbalance. We then train a two-tower model to generate applicant and job embeddings and compute their relevance scores. The resulting matrix records the protected attributes of $6{,}073$ job applicants and their relevance scores on $11{,}471$ jobs.
}

\subsection{Formal Proofs}\label{sec:proofs}

\begin{proof}[Proof of Lemma~\ref{lem:marginal}]
When allocation on partition $\M_\ell$, let $\U$ be the corresponding utility matrix after the $t$-th allocation and before the $(t+1)$-th allocation for $t\in [|\M_\ell|]$, and let $i,j\in [K]$ be the group indices. We assume $\U[j,\ell]\ge 1$ after the $t$-th allocation, \ie at least one item assigned to the $j$-th group. When $i=j$, it is trivial to know that $g(\U[j,\ell],\alpha_j) > g(\U[i,\ell]+1,\alpha_i)$ according to Proposition~\ref{pro:monotone}. When $i\neq j$, we consider two cases in terms of the iteration when the number of items assigned to the $j$-th group reaches $\U[j,\ell]$.

\spara{Case 1} The number of items assigned to the $j$-th group reaches $\U[j,\ell]$ at the $t$-th iteration. In this case, we know the number of items assigned to each group right before the $t$-th iteration are $\{\U[1,\ell],\cdots, \U[j-1,\ell], {\bf \U[j,\ell]-1}, \U[j+1,\ell], \cdots, \U[K,\ell]\}$. Then for the $t$-th iteration, we know that $g(\U[j,\ell],\alpha_j)=\max\{g(\U[1,\ell]+1,\alpha_1),\cdots,g(\U[j,\ell],\alpha_j),\cdots,g(\U[K,\ell]+1,\alpha_K)\}$. As a consequence, $g(\U[j,\ell],\alpha_j) \ge g(\U[i,\ell]+1,\alpha_i)$ for $\forall i,j\in [K]$ holds.

\spara{Case 2} The number of items assigned to the $j$-th group reaches $\U[j,\ell]$ at the $t^\prime$-th iteration with $t^\prime<t$. Let $\U^\prime$ be the utility matrix right before the $t^\prime$-th iteration. Thus we have $\U^\prime[j,\ell]=\U[j,\ell]-1$ and $\U^\prime[i,\ell] \le \U[i,\ell]$. Similarly, at this $t^\prime$-th iteration, we have $g(\U[j,\ell],\alpha_j)=\max\{g(\U^\prime[1,\ell]+1,\alpha_1),\cdots,g(\U[j,\ell],\alpha_j)),\cdots,g(\U^\prime[K,\ell]+1,\alpha_K)\}$. Therefore, we know $g(\U[j,\ell],\alpha_j) \ge  g(\U^\prime[i,\ell]+1,\alpha_i) \ge g(\U[i,\ell]+1,\alpha_i)$ according to Proposition~\ref{pro:monotone}. In particular, for the case of the last iteration of $t=|\M_\ell|$, the above two cases still apply, which completes the proof.
\end{proof}

\begin{proof}[Proof of Theorem~\ref{thm:optimalG}]

Since the welfare objective is additive over partitions, and the global capacity constraint is never binding under our setting as $\sum_{\ell}U[k,\ell]\le |N_k|$, it suffices to prove optimality independently for each fixed partition $\mathcal M_\ell$. Hence, fix an arbitrary partition $\M_\ell$ with $|\M_\ell|$ items. Since the welfare
$W(\U,\alpha)=\sum_{\ell\in[L]}\sum_{k\in[K]}\left(\frac{\U[k,\ell]}{|\M_\ell|}\right)^{\alpha_k}$
is separable across partitions and Algorithm~\ref{alg:RA} allocates items independently for each $\ell$, it suffices to prove that for this fixed $\ell$, utility matrix $\U[\cdot,\ell]$ returned by Algorithm~\ref{alg:RA} maximizes $\sum_{k\in[K]}\left(\frac{\U[k,\ell]}{|\M_\ell|}\right)^{\alpha_k}$ subject to $\sum_{k\in[K]} \U[k,\ell]=|\M_\ell|$ and $\U[k,\ell]\in\mathbb{N}_+$ (proved in Corollary~\ref{cor:nonempty}).

Let $\U[\cdot,\ell]$ be the final greedy allocation on partition $\M_\ell$. Lemma~\ref{lem:marginal} implies that at termination we have the discrete optimality condition
\begin{equation}\label{eq:terminal-g}
g(\U[j,\ell],\alpha_j)\ \ge\ g(\U[i,\ell]+1,\alpha_i),\qquad \forall i,j\in[K],
\end{equation}
\ie every \emph{last allocated} marginal gain is at least every \emph{next unallocated} marginal gain.

Consider any other feasible integer allocation $\U^\prime[\cdot, \ell]$ with $\sum_{k\in[K]} \U^\prime[k, \ell]=|\M_\ell|$. If $\U^\prime[\cdot, \ell]=\U[\cdot,\ell]$ we are done. Otherwise, since both sum to $|\M_\ell|$, there exist indices $i,j\in[K]$ such that $\U^\prime[i, \ell]\ge \U[i,\ell]+1$ and $\U^\prime[j, \ell]\le \U[j,\ell]-1$. Construct $\U^\ast[\cdot, \ell]$ by moving one item from group $i$ to group $j$: $\U^\ast[i,\ell]=\U^\prime[i,\ell]-1$, $\U^\ast[j,\ell]=\U^\prime[j,\ell]+1$, and $\U^\ast[k,\ell]=\U^\prime[k,\ell]$ for $k\notin\{i,j\}$. The welfare change on partition $\M_\ell$ equals $g(\U^\prime[j,\ell]+1,\alpha_j)-g(\U^\prime[i,\ell],\alpha_i)$. Using the diminishing-returns property of $g(\cdot,\alpha_k)$ and the relations
$\U^\prime[j,\ell]+1\le \U[j,\ell]$ and
$\U^\prime[i,\ell]\ge \U[i,\ell]+1$, we have
$g(\U^\prime[j,\ell]+1,\alpha_j)\ge g(\U[j,\ell],\alpha_j)$ and
$g(\U^\prime[i,\ell],\alpha_i)\le g(\U[i,\ell]+1,\alpha_i)$.
Therefore,
$g(\U^\prime[j,\ell]+1,\alpha_j)-g(\U^\prime[i,\ell],\alpha_i)
\ge g(\U[j,\ell],\alpha_j)-g(\U[i,\ell]+1,\alpha_i)\ge 0$,
where the last inequality follows from Eq.~\eqref{eq:terminal-g}.
Hence, the exchange from $\U^\prime[\cdot,\ell]$ to $\U^\ast[\cdot,\ell]$
does not decrease welfare.

Repeating the above exchange operation, we can transform $\U^\prime[\cdot,\ell]$ into $\U[\cdot,\ell]$ in finitely many steps, and welfare never decreases along the process. Thus the greedy allocation $\U[\cdot,\ell]$ attains welfare at least that of any feasible $\U^\prime[\cdot, \ell]$, \ie it is optimal for partition $\M_\ell$. Since partitions are independent and welfare is additive across $\ell$, the matrix $\U$
returned by Algorithm~\ref{alg:RA} maximizes $\W(\U,\alpha)$ for the given $\alpha$.
\end{proof}

\begin{proof}[Proof of Lemma~\ref{lem:fixportion}]

\arxiv{Given an initialized $\alpha_1\in (0,1), 0<\x_2\le \x_1<1$ , we prove the equation $\alpha_1 \x_1^{\alpha_1-1}=\alpha_2 \x_2^{\alpha_2-1}$ with $\alpha_2 \in (0,\alpha_1]$. By reorganizing equation $\alpha_1 \x_1^{\alpha_1-1}=\alpha_2 \x_2^{\alpha_2-1}$, we have 
\begin{align}
\alpha_1 \x_1^{\alpha_1-1} &= \alpha_2 \x_2^{\alpha_2-1}, & 
\tfrac{\alpha_1}{\x_1^{1-\alpha_1}} &=\tfrac{\alpha_2}{\x_2^{1-\alpha_2}}. \label{eqn:alphax}
\end{align}
First, it is clear that $\tfrac{\alpha_1}{\x_1^{1-\alpha_1}}>0$ and $\tfrac{\alpha_2}{\x_2^{1-\alpha_2}}>0$ for $\alpha_1, \alpha_2, \x_1, \x_2 \in (0,1)$. Let function $f(\alpha_2)=\tfrac{\alpha_2}{\x_2^{1-\alpha_2}}$ with $\alpha_2\in (0,\alpha_1]$. Since $\x_1 \ge \x_2$, we have $\tfrac{\alpha_1}{\x_1^{1-\alpha_1}}\le \tfrac{\alpha_1}{\x_2^{1-\alpha_1}}=f(\alpha_1)$ by setting $\alpha_2=\alpha_1$. Meanwhile, it is trivial that $f(\alpha_2) \rightarrow 0$ when $\alpha_2 \rightarrow 0$.

Next, we prove that $f(\alpha_2)$ is continuous on $(0,\alpha_1]$ such that there exists $\alpha^\ast \in (0,\alpha_1]$ to ensure $f(\alpha^\ast)=\tfrac{\alpha_1}{\x_1^{1-\alpha_1}}$. Equivalently, we rewrite $f(\alpha_2)=\tfrac{\alpha_2}{\x_2^{1-\alpha_2}}$ as $f(\alpha_2)=\alpha_2\e^{(\alpha_2-1)\ln{x_2}}$. We know that on $(0,\alpha_1]$, i) function $f_1(\alpha_2)=\alpha_2$ is continuous, ii) function $f_2(\alpha_2)=\alpha_2-1$ is continuous, and iii) the fact that $\e^{\alpha_2}$ is continuous indicates $f_3(\alpha_2)=\e^{(\alpha_2-1)\ln{x_2}}$ is continuous as the combination of continuous functions is continuous. Therefore,  $f(\alpha_2)=f_1(\alpha)f_3(\alpha_2)$ is continuous on $(0,\alpha_1]$. Note that $f(\alpha_2)$ is continuous on $(0,\alpha_1]$, and
$\lim_{\alpha_2\to 0^+} f(\alpha_2)=0, f(\alpha_1)=\frac{\alpha_1}{x_2^{1-\alpha_1}}
\ge \frac{\alpha_1}{x_1^{1-\alpha_1}}$. Therefore, by the intermediate value theorem, there exists $\alpha^\ast \in (0,\alpha_1]$ to ensure $\tfrac{\alpha^\ast}{\x_2^{1-\alpha^\ast}}=\tfrac{\alpha_1}{\x_1^{1-\alpha_1}}$, which completes the proof.}
\sketch{The proof sketch is to transform the equation
$\alpha_1 \x_1^{\alpha_1-1}=\alpha_2 \x_2^{\alpha_2-1}$
into $\tfrac{\alpha_1}{\x_1^{1-\alpha_1}}=\tfrac{\alpha_2}{\x_2^{1-\alpha_2}}$, and define the function
$f(\alpha_2)=\tfrac{\alpha_2}{\x_2^{1-\alpha_2}}$
over $\alpha_2\in(0,\alpha_1]$.
Since $\x_1\ge\x_2$, we have $f(\alpha_1)=\tfrac{\alpha_1}{\x_2^{1-\alpha_1}} \ge \tfrac{\alpha_1}{\x_1^{1-\alpha_1}}$. Meanwhile, $f(\alpha_2)\to 0$ as $\alpha_2\to0^+$. Because $f(\alpha_2)$ is continuous on $(0,\alpha_1]$, the intermediate value theorem guarantees the existence of some $\alpha^\ast\in(0,\alpha_1]$ such that
$f(\alpha^\ast)=\tfrac{\alpha_1}{\x_1^{1-\alpha_1}}$,
which proves the claim.}
\end{proof}

\begin{lemma}\label{lem:g2bound}
Let $M\ge 2$ be an integer. For any constant $\lambda\in(0,1/M]$, it holds that $g(2,\lambda) = (2/M)^{\lambda}-(1/M)^{\lambda} < (1/M)^{1/M}$.
\end{lemma}

\arxiv{
\begin{proof}[Proof of Lemma~\ref{lem:g2bound}]
We rewrite
\[
g(2,\lambda)=(1/M)^{\lambda}\,(2^{\lambda}-1)\le 2^{\lambda}-1\le 2^{1/M}-1,
\]
where the first inequality follows from $(1/M)^{\lambda}\le 1$ and the second from $\lambda\le 1/M$.

Using the standard inequality $e^x-1\le x e^x$ for all $x\ge 0$ and setting
$x=\ln 2/M$, we obtain
\[
2^{1/M}-1 = e^{(\ln 2)/M}-1 \le \frac{\ln 2}{M}\,e^{(\ln 2)/M}
= \frac{\ln 2}{M}\,2^{1/M}.
\]
Since $M\ge 2$, we have $2^{1/M}\le \sqrt{2}$ and $\ln 2/M\le \ln 2/2$, which implies $2^{1/M}-1 \le \frac{\ln 2}{2}\sqrt{2} < \tfrac{1}{2}$. On the other hand, $(1/M)^{1/M}$ is minimized over integers $M\ge 2$ at $M=3$, where $(1/3)^{1/3}>\tfrac{1}{2}$. Therefore, $g(2,\lambda)\le 2^{1/M}-1 < (1/M)^{1/M}$, which completes the proof.
\end{proof}
}

\begin{proof}[Proof of Corollary~\ref{cor:nonempty}]
Fix any $\ell\in[L]$ and let $M=|\M_\ell|$. By Algorithm~\ref{alg:alphas}, we have
$\alpha_1\ge \alpha_2\ge \cdots \ge \alpha_K$. With the initialization
$\alpha_1=1/M_{\max}$ where $M_{\max}=\max\{|\M_1|,\ldots,|\M_L|\}$, it follows that
$\alpha_k\le 1/M_{\max}\le 1/M$ for all $k\in[K]$. 

Recall the marginal gain function on partition $\M_\ell$ is $g(t,\alpha_k) = (t/M)^{\alpha_k} - ((t-1)/M)^{\alpha_k}$ for $t\in\{1,2,\ldots,M\}$. Hence, for any already-allocated group with $\U[k,\ell]\ge 1$ (\ie $t\ge 1$), its next-step marginal gain satisfies $g(t+1,\alpha_k)\le g(2,\alpha_k) = (2/M)^{\alpha_k}-(1/M)^{\alpha_k}$. According to Lemma~\ref{lem:g2bound}, we have $g(2,\alpha_k) = (2/M)^{\alpha_k}-(1/M)^{\alpha_k} < (1/M)^{1/M}$ for $\alpha_k\in(0,1/M]$.

On the other hand, for any group $j$ that has received no item yet, \ie $\U[j,\ell]=0$,
its first-item marginal gain is
$g(1,\alpha_j) = (1/M)^{\alpha_j} \ge (1/M)^{1/M}$,
because $\alpha_j\le 1/M$ and $1/M\in(0,1)$. Consequently, as long as there exists a group with $\U[\cdot,\ell]=0$, every already-allocated group
has next-step marginal gain $< (1/M)^{1/M}$, while each unallocated group has first-item marginal gain
$\ge (1/M)^{1/M}$. Hence, \ours must allocate to an unallocated group at each iteration until all $K$
groups receive at least one item. Since $M> K$ by assumption, this is feasible and yields
$\U[k,\ell]\ge 1$ for all $k\in[K]$. As $\ell$ is arbitrary, the claim holds for all $\ell\in[L]$.
\end{proof}

\begin{proof}[Proof of Lemma~\ref{lem:similarmargin}]
Fix any partition $\ell\in[L]$ and let $M=|\M_\ell|$. Let $\x=(\x_1,\ldots,\x_K)$ be the optimal solution of the divisible relaxation on $\M_\ell$. By Eq.~\ref{eqn:marginalequation}, there exists a common constant $\lambda>0$ such that $f_x^\prime(\x_k,\alpha_k)=\lambda$ for all $k\in[K]$, where $f(x,\alpha)=x^\alpha$ and $f_x^\prime(x,\alpha)=\alpha x^{\alpha-1}$.
For the indivisible allocation $\U$ produced by \ours, recall the marginal gain
$g(t,\alpha_k)=\W(t,\alpha_k)-\W(t-1,\alpha_k)$, which admits the integral form $g(t,\alpha_k)=\int_{(t-1)/M}^{t/M} f_x^\prime(x,\alpha_k)\,dx$. Since $f(\cdot,\alpha_k)$ is concave for $\alpha_k\in(0,1)$, $f_x^\prime(x,\alpha_k)$ is strictly decreasing in $x$. Define $b_k=\lfloor \x_k M\rfloor$ for $k\in[K]$ and let $R=M-\sum_{k\in[K]} b_k$, so $0\le R\le K-1$. We now prove that for every group $k$, the integer allocation $\U[k, \ell]$ differs from $\lfloor \x_k M \rfloor$ by at most $\mathcal{O}(K)$ via a contradiction argument based on the greedy marginal ordering.

\paragraph{Lower bound.}
We show $\U[k,\ell]\ge \max\{b_k-K,0\}+1$ for all $k\in[K]$.
Suppose for contradiction that there exists $i\in[K]$ such that $\U[i,\ell]\le b_i-K$.
Then $b_i-\U[i,\ell]\ge K$, and hence the total surplus of the remaining $K-1$ groups relative to $(b_k)_{k\in[K]}$ satisfies
$\sum_{k\ne i}(\U[k,\ell]-b_k)=R+(b_i-\U[i,\ell])\ge R+K\ge K$.
Since there are only $K-1$ groups in $[K]\setminus\{i\}$, by pigeonhole there exists some $j\ne i$ with $\U[j,\ell]\ge b_j+2$; otherwise each $k\ne i$ would satisfy $\U[k,\ell]\le b_k+1$ and the total surplus would be at most $K-1$, contradicting $\sum_{k\ne i}(\U[k,\ell]-b_k)\ge K$.

Now compare the next unallocated marginal of group $i$ and the last allocated marginal of group $j$.
Since $\U[i,\ell]+1\le b_i-K+1\le b_i$ and $b_i/M\le \x_i$, the interval $[(\U[i,\ell])/M,(\U[i,\ell]+1)/M]$ lies strictly to the left of (or ends at) $\x_i$, and thus $f_x^\prime(x,\alpha_i)>\lambda$ for all $x$ in this interval except possibly at a single endpoint; hence $g(\U[i,\ell]+1,\alpha_i)>\lambda/M$.
Similarly, since $\U[j,\ell]\ge b_j+2$ and $\x_j<(b_j+1)/M$, the interval $[(\U[j,\ell]-1)/M,\U[j,\ell]/M]$ lies strictly to the right of $\x_j$, and thus $f_x^\prime(x,\alpha_j)<\lambda$ throughout; hence $g(\U[j,\ell],\alpha_j)<\lambda/M$.
Therefore, $g(\U[i,\ell]+1,\alpha_i)>g(\U[j,\ell],\alpha_j)$, contradicting the greedy terminal optimality condition of \ours on $\M_\ell$ according to Lemma~\ref{lem:marginal}.
This proves $\U[i,\ell]\ge b_i-K+1$, and hence $\U[k,\ell]\ge \max\{b_k-K,0\}+1$ for all $k$.

\paragraph{Upper bound.}
We show $\U[k,\ell]\le \lceil \x_k M\rceil+K-1$ for all $k\in[K]$.
Suppose for contradiction that there exists $i\in[K]$ with $\U[i,\ell]\ge \lceil \x_i M\rceil+K$.
If $\x_i M\notin\mathbb{Z}$, then $\lceil \x_i M\rceil=b_i+1$ and thus $\U[i,\ell]\ge b_i+K+1$, implying $\U[i,\ell]-b_i\ge K+1$.
If $\x_i M\in\mathbb{Z}$, then $\lceil \x_i M\rceil=b_i$ and thus $\U[i,\ell]\ge b_i+K$, implying $\U[i,\ell]-b_i\ge K$.
In both cases, we have $\U[i,\ell]-b_i\ge K$.
Since the total surplus relative to $(b_k)$ equals $R\le K-1$, there must exist some $j\ne i$ with $\U[j,\ell]\le b_j-1$; otherwise every $k\ne i$ would satisfy $\U[k,\ell]\ge b_k$ and the total surplus would be at least $K$, contradicting $R\le K-1$.

Consider the last allocated marginal of group $i$ and the next unallocated marginal of group $j$.
Because $\U[i,\ell]\ge b_i+K\ge b_i+2$ due to $(K\ge 2)$ and $\x_i<(b_i+1)/M$, the interval $[(\U[i,\ell]-1)/M,\U[i,\ell]/M]$ lies strictly to the right of $\x_i$, hence $g(\U[i,\ell],\alpha_i)<\lambda/M$.
Because $\U[j,\ell]\le b_j-1$, we have $\U[j,\ell]+1\le b_j$ and $b_j/M\le \x_j$, so the interval $[(\U[j,\ell])/M,(\U[j,\ell]+1)/M]$ lies to the left of (or ends at) $\x_j$, hence $g(\U[j,\ell]+1,\alpha_j)>\lambda/M$. Therefore, $g(\U[i,\ell],\alpha_i)<g(\U[j,\ell]+1,\alpha_j)$, again contradicting the greedy terminal optimality condition of \ours on $\M_\ell$. This proves $\U[i,\ell]\le \lceil \x_i M\rceil+K-1$, and hence $\U[k,\ell]\le \lceil \x_k M\rceil+K-1$ for all $k$.

Combining the two parts yields the stated bounds for all $k\in[K]$ on partition $\M_\ell$. As $\ell$ is arbitrary, the lemma holds for all $\ell\in[L]$.
\end{proof}

\begin{proof}[Proof of Theorem~\ref{thm:worstinequality}]

\arxiv{Fix a partition $\M_\ell$ and let $M=|\M_\ell|$. For each group $k\in[K]$, define
$b_k=\lfloor \x_k M\rfloor$ and denote $\U_k=\U[k,\ell]$.
For $\beta\in(0,1)$, let $p=1-\beta\in(0,1)$ and define the per-agent quantity
$\b_k=\U_k/(M|\N_k|)$.

Recall $f(x,\alpha)=x^\alpha$ and the discrete marginal gain
$g(t,\alpha)=f(t/M,\alpha)-f((t-1)/M,\alpha)
=\int_{(t-1)/M}^{t/M} f_x'(x,\alpha)\,dx$,
where $f_x'(x,\alpha)=\alpha x^{\alpha-1}$ is strictly decreasing in $x$ for $\alpha\in(0,1)$.
By the KKT condition of the divisible optimum $\x$
(Equation~\eqref{eqn:marginalequation} / Lemma~\ref{lem:kkt}),
there exists $\lambda>0$ such that $f_x'(\x_k,\alpha_k)=\lambda$ for all $k\in[K]$.

\smallskip
\noindent\textbf{Two exclusion relations.}

\begin{claim}\label{clm:exclusion1}
If there exists some $i\in[K]$ such that $\U_i\ge b_i+2$, then $\U_j\ge b_j$ holds for all $j\in[K]$.
\end{claim}
\begin{proof}
Suppose for contradiction that there exist $i,j$ such that $\U_i\ge b_i+2$ and $\U_j\le b_j-1$.
Since $\U_j+1\le b_j$ and $b_j/M\le \x_j$, the interval
$[\U_j/M,(\U_j+1)/M]$ lies to the left of (or ends at) $\x_j$.
Thus $f_x'(x,\alpha_j)>\lambda$ throughout this interval, and hence
$g(\U_j+1,\alpha_j)>\lambda/M$.

On the other hand, $\U_i\ge b_i+2$ implies $\U_i-1\ge b_i+1$, and since $\x_i<(b_i+1)/M$,
the interval $[(\U_i-1)/M,\U_i/M]$ lies strictly to the right of $\x_i$.
Thus $f_x'(x,\alpha_i)<\lambda$ on this interval and hence
$g(\U_i,\alpha_i)<\lambda/M$. Therefore $g(\U_j+1,\alpha_j)>g(\U_i,\alpha_i)$,
contradicting the terminal optimality condition of \ours (Lem~\ref{lem:marginal}, condition T1).
\end{proof}

\begin{claim}\label{clm:exclusion2}
If there exists some $i\in[K]$ such that $\U_i\le b_i-1$, then $\U_j\le b_j+1$ holds for all $j\in[K]$.
\end{claim}
\begin{proof}
Assume $\U_i\le b_i-1$ for some $i$.
If there existed $j$ with $\U_j\ge b_j+2$, then Claim~\ref{clm:exclusion1} would imply $\U_i\ge b_i$,
a contradiction.
Hence $\U_j<b_j+2$ for all $j$, and since $\U_j,b_j$ are integers, $\U_j\le b_j+1$.
\end{proof}

\smallskip
\noindent\textbf{Normalization and surplus.}
Let $\mu=\frac{1}{K}\sum_{k=1}^K \b_k$ and $\y=\b/\mu$.
By scale invariance, $A(\b,\beta)=A(\y,\beta)$, and
$A(\y,\beta)=1-\big(\frac{1}{K}\sum_{k=1}^K \y_k^p\big)^{1/p}$.
Define the surplus amount
$R:=M-\sum_{k=1}^K b_k=M-\sum_{k=1}^K\lfloor \x_k M\rfloor\in[0,K-1]$.

\smallskip
\noindent\emph{Case A: there exists $i$ such that $\U_i\ge b_i+2$.}
By Claim~\ref{clm:exclusion1}, $\U_k\ge b_k$ for all $k$.
Let $s_k=\U_k-b_k\in\Z_{\ge0}$; then $\sum_k s_k=R$.

\smallskip
\noindent\textbf{Exchange argument and extremal allocation.}
Assume $|\N_1|\le \cdots \le |\N_K|$.
For any feasible surplus vector $s$ with $\sum_k s_k=R$,
if there exist $a<b$ with $s_b\ge1$, define $s'$ by
$s'_a=s_a+1$, $s'_b=s_b-1$, and $s'_k=s_k$ otherwise.
Let $\U'$ be the corresponding allocation and define $\b',\mu',\y'$ analogously.

\begin{lemma}\label{lem:exchange_majorization}
Let $\y\in\R_+^K$ satisfy $\frac{1}{K}\sum_{k=1}^K \y_k=1$ and let $p\in(0,1)$.
Suppose $\y'$ is obtained from $\y$ by increasing one coordinate by $\Delta_a$,
decreasing another by $\Delta_b$, with $\Delta_a\ge\Delta_b\ge0$,
and then rescaling all coordinates by a common factor so that the mean remains $1$. Consequently we have
$\sum_{k=1}^K (\y'_k)^p \le \sum_{k=1}^K \y_k^p$.
\end{lemma}

\begin{proof}[Proof of Lemma~\ref{lem:exchange_majorization}]
Let $\z$ denote the unnormalized vector after the two-coordinate update, and let
$\rho=\frac{1}{K}\sum_k \z_k\ge1$ be the rescaling factor.
All coordinates except the increased one are scaled down by $\rho$,
while one coordinate gains additional mass before rescaling.
Thus, after sorting in nonincreasing order, for every $m\in[K]$
the sum of the $m$ largest coordinates of $\y'=\z/\rho$ is at least that of $\y$,
with equality at $m=K$.
Since $t\mapsto t^p$ is concave for $p\in(0,1)$,
Karamata's inequality yields the desired inequality.
\end{proof}

In our setting, the update satisfies the conditions of Lemma~\ref{lem:exchange_majorization}
with $\Delta_a=\frac{1}{M|\N_a|\mu}$ and $\Delta_b=\frac{1}{M|\N_b|\mu}$,
and $\Delta_a\ge\Delta_b$ because $|\N_a|\le|\N_b|$.
Therefore $A(\y',\beta)\ge A(\y,\beta)$.

Iterating the exchange moves all surplus to group $1$ and yields the extremal allocation
$\tilde{\U}_1=b_1+R$ and $\tilde{\U}_k=b_k$ for $k\ge2$.
Hence for any $\U$ in Case A, $A(\b,\beta)\le A(\tilde{\b},\beta)$, where
$\tilde{\b}_1=\frac{b_1+R}{M|\N_1|}$ and $\tilde{\b}_k=\frac{b_k}{M|\N_k|}$.

\smallskip
\noindent\emph{Case B: for all $i$, $\U_i\le b_i+1$.}
Then $s_k=\U_k-b_k\in\{0,1\}$ and $\sum_k s_k=R$.
If there exist $a<b$ with $s_a=0$ and $s_b=1$, the same exchange argument applies and does not decrease
the Atkinson inequality.
Iterating again yields the same extremal allocation $\tilde{\U}$.
Thus $A(\b,\beta)\le A(\tilde{\b},\beta)$ also holds in Case B.

Combining the two cases, we conclude that for any allocation $\U$ returned by \ours on $\M_\ell$, its Atkinson inequality is maximized when all $R$ surplus items are assigned to the group with the smallest population size, namely $\N_1$. In this extremal allocation, $\U_1=\lfloor \x_1 M\rfloor+R$ and $\U_k=\lfloor \x_k M\rfloor$ for all $k\ge2$. Substituting this allocation into the definition of the Atkinson inequality yields
\[\textstyle A(\b,\beta)\le1-\frac{1}{\mu}\left(\frac{1}{K}\left(\frac{\lfloor \x_1 M\rfloor+R}{M|\N_1|}\right)^{1-\beta}+\frac{1}{K}\sum_{k=2}^K\left(\frac{\lfloor \x_k M\rfloor}{M|\N_k|}\right)^{1-\beta}\right)^{\frac{1}{1-\beta}},\]
where $\mu=\frac{1}{K}\left(\frac{\lfloor \x_1 M\rfloor+R}{M|\N_1|} +\sum_{k=2}^K\frac{\lfloor \x_k M\rfloor}{M|\N_k|}\right)$, which completes the proof. 
}\sketch{The proof sketch consists of two main steps. First, let $b_k=\lfloor x_k M\rfloor$ and denote the surplus amount by
$R=M-\sum_{k\in[K]} b_k$.
Using the KKT condition of the divisible optimum $\x$, there exists $\lambda>0$ such that
$f_x'(\x_k,\alpha_k)=\lambda$ for all $k\in[K]$.
Since the marginal gain function
$g(t,\alpha)=f(t/M,\alpha)-f((t-1)/M,\alpha)$
is strictly decreasing in $t$, we obtain two exclusion relations:
i) if some group satisfies $\U_i\ge b_i+2$, then every other group must satisfy $\U_j\ge b_j$;
ii) if some group satisfies $\U_i\le b_i-1$, then every other group must satisfy $\U_j\le b_j+1$.
Consequently, the integer allocation returned by \ours differs from the divisible optimum only through the redistribution of the surplus amount $R$.

Second, define the normalized per-agent allocation
$\b_k=\U_k/(M|\N_k|)$ and let
$A(\b,\beta)=1-(\frac{1}{K}\sum_k \y_k^{1-\beta})^{1/(1-\beta)}$
with $\y=\b/\mu$.
Since $t^{1-\beta}$ is concave for $\beta\in(0,1)$, an exchange argument together with Karamata's inequality shows that moving one surplus item from a larger group to a smaller group cannot decrease the Atkinson inequality.
Therefore, repeatedly applying such exchanges yields the extremal allocation in which all surplus items are concentrated in the smallest group:
$\tilde{\U}_1=b_1+R$ and $\tilde{\U}_k=b_k$ for $k\ge2$.
Substituting this extremal allocation into the definition of the Atkinson inequality gives the stated upper bound.
}
\end{proof}

\begin{proof}[Proof of Theorem~\ref{thm:oursa}]
Fix any partition $\ell\in[L]$ and denote $M=|\M_\ell|$. Since the welfare is additive over partitions and each constraint $\C_\ell$ is group-wise separable, it suffices to show that for each fixed $\ell$, the vector $\U[\cdot,\ell]$ returned by \oursa maximizes the partition welfare $\sum_{k\in[K]}\W(\U[k,\ell],\alpha_k)$ subject to $\sum_{k\in[K]}\U[k,\ell]\le M$ and
$0\le \U[k,\ell]\le c_{k,\ell}$.

We first state a greedy optimality condition induced by Algorithm~\ref{alg:RAC} on partition $\M_\ell$. At any successful assignment step, Algorithm~\ref{alg:RAC} selects a feasible group $j$ such that $g(\U[j,\ell]+1,\alpha_j)\ge g(\U[i,\ell]+1,\alpha_i)$ holds for every group $i$ with $\U[i,\ell]+1\le c_{i,\ell}$. Moreover, if Algorithm~\ref{alg:RAC} does not assign all $M$ items on partition $\M_\ell$, then it must terminate because no group remains feasible, i.e., $\U[i,\ell]=c_{i,\ell}$ holds for all $i\in[K]$.

Let $\U^\circ[\cdot,\ell]$ be a feasible allocation on $\M_\ell$.
Assume for contradiction that
$\sum_{k\in[K]} \W(\U[k,\ell],\alpha_k) < \sum_{k\in[K]} \W(\U^\circ[k,\ell],\alpha_k)$.
We consider two cases.

\spara{Case one: $\sum_{k\in[K]} \U[k,\ell] < M$} In this case, Algorithm~\ref{alg:RAC} must have terminated early. As noted above, this implies $\U[i,\ell]=c_{i,\ell}$ for all $i\in[K]$.
Since $\U^\circ$ is feasible under the same box constraints, it follows that
$\U^\circ[i,\ell]\le c_{i,\ell}=\U[i,\ell]$ for all $i\in[K]$.
Therefore $\U^\circ$ cannot assign more items or achieve strictly larger welfare than $\U$,
contradicting the assumption.

\spara{Case two: $\sum_{k\in[K]} \U[k,\ell] = M$}
If $\U^\circ[\cdot,\ell]=\U[\cdot,\ell]$, we are done. Otherwise, since both allocations assign
exactly $M$ items, there exist indices $i,j\in[K]$ such that
$\U^\circ[i,\ell]\ge \U[i,\ell]+1$ and $\U^\circ[j,\ell]\le \U[j,\ell]-1$.
Then $\U[i,\ell]+1\le \U^\circ[i,\ell]\le c_{i,\ell}$, so increasing group $i$ by one is feasible,
and $\U[j,\ell]\ge 1$.

Consider the feasible exchange that moves one item from $j$ to $i$, producing $\U'[\cdot,\ell]$ with
$\U'[i,\ell]=\U[i,\ell]+1$, $\U'[j,\ell]=\U[j,\ell]-1$, and $\U'[k,\ell]=\U[k,\ell]$ for
$k\notin\{i,j\}$.
The welfare change equals
$g(\U[i,\ell]+1,\alpha_i)-g(\U[j,\ell],\alpha_j)$.

Let $t$ be the iteration at which Algorithm~\ref{alg:RAC} assigns the last item to group $j$ on
partition $\M_\ell$, and let $\U^{(t-1)}$ be the allocation state immediately before that assignment.
Then $\U^{(t-1)}[j,\ell]+1=\U[j,\ell]$ and $\U^{(t-1)}[i,\ell]\le \U[i,\ell]$.
By Proposition~\ref{pro:monotone},
$g(\U^{(t-1)}[i,\ell]+1,\alpha_i)\ge g(\U[i,\ell]+1,\alpha_i)$ and
$g(\U^{(t-1)}[j,\ell]+1,\alpha_j)=g(\U[j,\ell],\alpha_j)$.
Since group $i$ is feasible at state $\U^{(t-1)}$, the greedy choice at iteration $t$ implies $g(\U[j,\ell],\alpha_j) = g(\U^{(t-1)}[j,\ell]+1,\alpha_j) \ge g(\U^{(t-1)}[i,\ell]+1,\alpha_i) \ge g(\U[i,\ell]+1,\alpha_i)$. Hence the exchange does not increase welfare. Repeating such exchanges finitely many times transforms $\U[\cdot,\ell]$ into $\U^\circ[\cdot,\ell]$ while never increasing welfare, contradicting the assumption that $\sum_k \W(\U[k,\ell],\alpha_k) < \sum_k \W(\U^\circ[k,\ell],\alpha_k)$.

In both cases we obtain a contradiction. Hence $\U[\cdot,\ell]$ is welfare-optimal on $\M_\ell$. As $\ell$ is arbitrary and welfare is additive across partitions, the allocation returned by \oursa maximizes the total welfare under the constraints $\C$.
\end{proof}

\end{sloppy}
\end{document}